\documentclass[11pt]{imsart}
\usepackage{amsmath,amssymb,graphicx,amsthm}
\RequirePackage[authoryear,round]{natbib}
\RequirePackage[colorlinks,citecolor=blue,urlcolor=blue]{hyperref}

\usepackage{color}
\usepackage{wasysym}
\usepackage{caption}

\startlocaldefs

\newtheorem{theorem}{Theorem}
\newtheorem{lemma}[theorem]{Lemma}

\DeclareMathOperator*{\argmin}{argmin}

\newenvironment{enum}{
\begin{enumerate}
  \setlength{\itemsep}{1pt}
  \setlength{\parskip}{0pt}
  \setlength{\parsep}{0pt}
}{\end{enumerate}}

\endlocaldefs

\usepackage{sectsty}
\allsectionsfont{\bfseries\sffamily}

\let\hat\widehat
\let\tilde\widetilde

\begin{document}

\begin{center} 
{{\textbf{\Large \textsf{Analysis of a Mode Clustering Diagram}}}}
\end{center}

\begin{center}
{\large{

\begin{tabular}{ccccc}
Isabella Verdinelli$^\dagger$ and Larry Wasserman$^{\dagger,\ast}$  \\
\end{tabular}
\vspace*{.1in}

\begin{center}
Department of Statistics and Data Science$^{\dagger}$ \\
and Machine Learning Department$^{\ast}$
\end{center}

\begin{center}
Carnegie Mellon University
\end{center}

\vspace*{-.2in}

\begin{tabular}{c}
{\texttt{$\{$isabella,larry$\}$@stat.cmu.edu}}
\end{tabular}
}}
\end{center}

\vspace*{-.2in}

\begin{quote}
{\em Mode-based clustering methods
define clusters to be the basins of attraction
of the modes
of a density estimate.
The most common version is mean shift clustering
which uses a gradient ascent algorithm to find the basins.
\cite{rodriguez2014clustering}
introduced a new method 
that is faster and simpler than mean shift clustering.
Furthermore, they define a clustering diagram that provides a simple,
two-dimensional summary
of the mode clustering information.
We study the statistical properties of this diagram
and we propose some improvements and extensions.
In particular, we show a connection between the diagram and robust linear regression.}
\end{quote}

\section{Introduction}

Mode-based clustering methods
define clusters in terms of the modes 
of the density function.
For example,
the mean-shift clustering method
\citep{comaniciu2002mean, cheng1995mean}
defines the clusters to be the basins of attraction of each mode.
Specifically, if we take any point $x$ and follow the path of steepest
ascent of the density, then we end up at a mode.
This assigns every point to a mode
which forms thus a partition of the space.
In practice, the density is estimated using a kernel 
density estimate.
The mean shift algorithm
then approximates the steepest ascent paths.

\citep{rodriguez2014clustering}
introduced a new approach to mode-based clustering
that avoids iterative computation of the density estimator.
Furthermore,
they define
a diagram --- which we call the
{\em mode clustering diagram} ---
that provides a useful summary of the clustering information.
The diagram is simply a plot of the pairs
$(p(X_i),\delta(X_i))$
where $p(X_i)$ is the density
of the $i^{\rm th}$ point and
$\delta(X_i)$ is the distance to the nearest
neighbor with higher density.
Modes appear as isolated points in the top right of the diagram.
See Figure \ref{fig::FIG3} for a simple example.

\cite{rodriguez2014clustering} appeared in {\em Science}
and has received over 1,000 citations but, to the best of our knowledge,
has not been examined in the statistics literature.
In this paper, we study the properties of the mode diagram.
These properties then suggest a heuristic for deciding which points are modes.
Specifically, if we perform a robust linear regression of
$\log \delta(X_i)$ on 
$\log p(X_i)$ 
then modes correspond to large, positive outliers.

\begin{center}
\includegraphics[scale=.7]{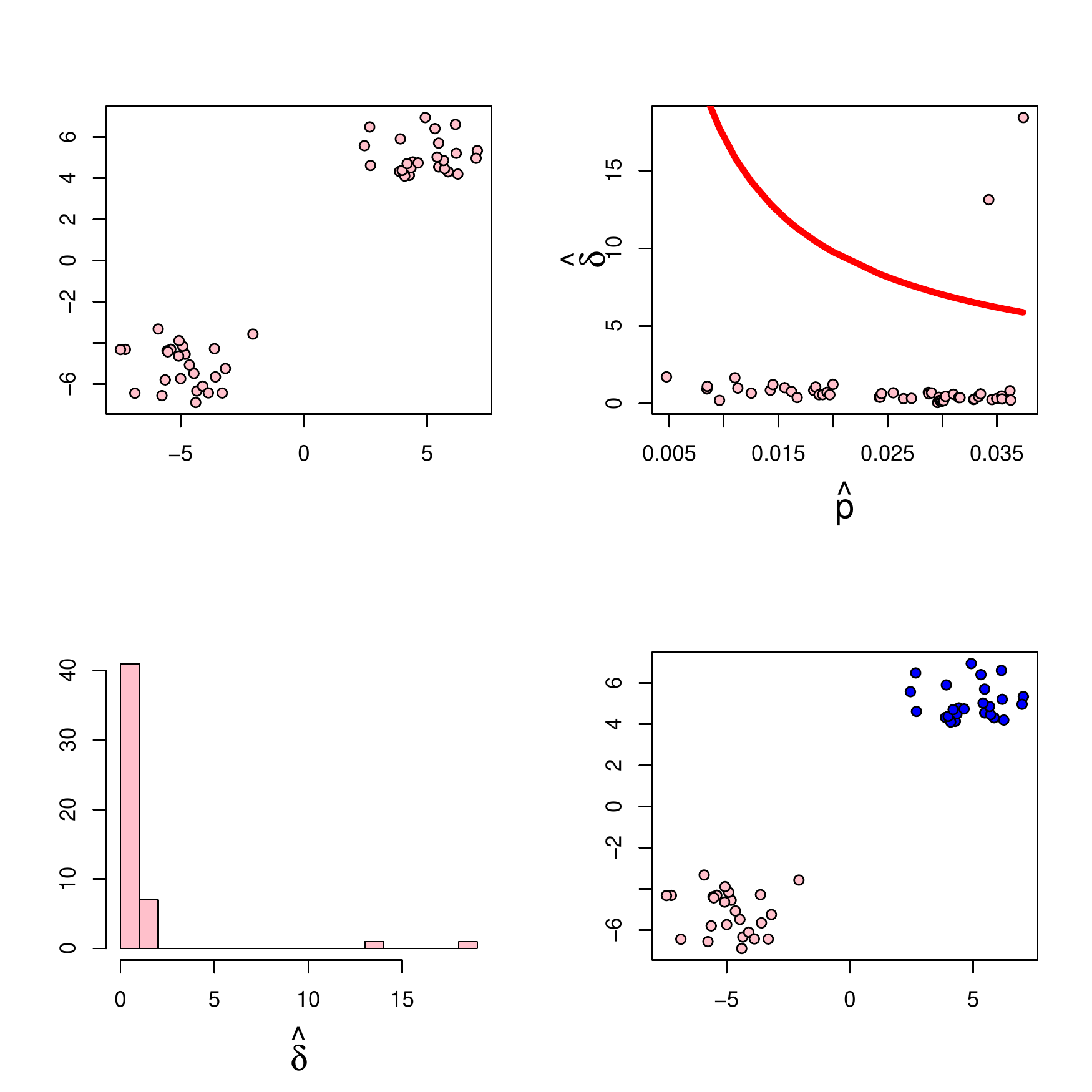}
\end{center}
\begin{center}
\vspace{-.2in}
\captionof{figure}{\em Top left: data.
Top right: the mode diagram.
The curved line is the threshold function $t_n$
corresponding to a robust linear regression 
of $\log \hat\delta$ on $\log \hat p$.
Points above the function are declared to be modes.
Bottom left: Histogram of $\hat \delta(X_i)$.
Bottom right: the resulting clusters.}
\label{fig::FIG3}
\end{center}

\subsection{Related Work}

The most common mode-based
clustering method is mean-shift clustering,
developed by
\cite{cheng1995mean} and
\cite{comaniciu2002mean}.
The method has been developed
in the statistics literature by
\cite{li2007nonparametric,
arias2015estimation,
chacon2015population,
chacon2013data,
chacon} and
\cite{genovese2016non}.

The new method --- the subject of this paper ---
is due to
\cite{rodriguez2014clustering}.
Many extensions and improvements
have since been proposed.
These extensions include speedups and methods for dealing with higher
dimensional problems.
Three highly cited such papers are \cite{wang2015fast,
du2016study,
courjault2016improved}.
In addition, there are now hundreds of 
scientific papers that apply the method to various applications.

\subsection{Paper Outline}

In Section
\ref{section::notation}
we establish the notation and the assumptions.
We review mode-based clustering in
Section \ref{section::clustering}.
We establish the theoretical properties of the population version
of the mode diagram in Section \ref{section::oracle}.
We then consider the estimated diagram
in Section \ref{section::not-oracle}.
Based on these results,
we suggest a method for thresholding the diagram
in Section \ref{section::threshold}.
In Section \ref{section::examples}
we illustrate the method with several examples.
Section \ref{section::conclusion} contains some concluding remarks.
All proofs are in the appendix.

\section{Notation and Assumptions}
\label{section::notation}

Let $X_1,\ldots, X_n$ be a sample from a distribution $P$
on $\mathbb{R}^d$.
We make the following assumptions throughout the paper:

(A1) $P$ is supported on a compact set ${\cal C}$ and
has bounded, continuous density $p$. Also, $\inf_{x\in {\cal C}}p(x) \geq  a > 0$.

(A2) $p$ has bounded and continuous 
first, second and third derivatives.
We let $g$ denote the gradient and we let $H$ denote the Hessian.

Recall that $x$ is a critical point if $||g(x)||=0$.
A function is {\em Morse}
\citep{milnor2016morse} if the Hessian is non-degenerate 
at every critical point.

(A3) $p$ is Morse with finitely many critical points.

The Morse assumption is critical to our proofs.
It may be possible to drop this assumption but the 
proof techniques would
have to change considerably.
The assumption that $\inf_{x \in {\cal X}}p(x) \geq a >0$
is not critical
and could be dropped at the expense of more 
involved statements and proofs.
More specifically,
the proofs then require dividing the sample space into
two regions: the first where
$p(x) \geq n^{-\frac{1}{d+2}}$ and the second where
$p(x) < n^{-\frac{1}{d+2}}$.
Also, points with
$\hat p(x) < n^{-\frac{1}{d+2}}$ should be removed from 
the cluster diagram.

\bigskip
A point $x$ is a mode if
there exists an $\epsilon>0$ and a ball
$B(x,\epsilon)$ such that
$p(x) > p(y)$ for all $y\in B(x,\epsilon)$, $y\neq x$.
Let ${\cal M} =\{m_1,\ldots, m_k\}$ denote the modes.
Because $p$ is Morse,
$x$ is a mode
if and only if
$g(x) = (0,\ldots, 0)^T$ and
$\lambda_{\max}(H(x)) < 0$
where $\lambda_{\max}(A)$ denotes the 
largest eigenvalue of a matrix $A$.

\begin{center}
\includegraphics[scale=.5]{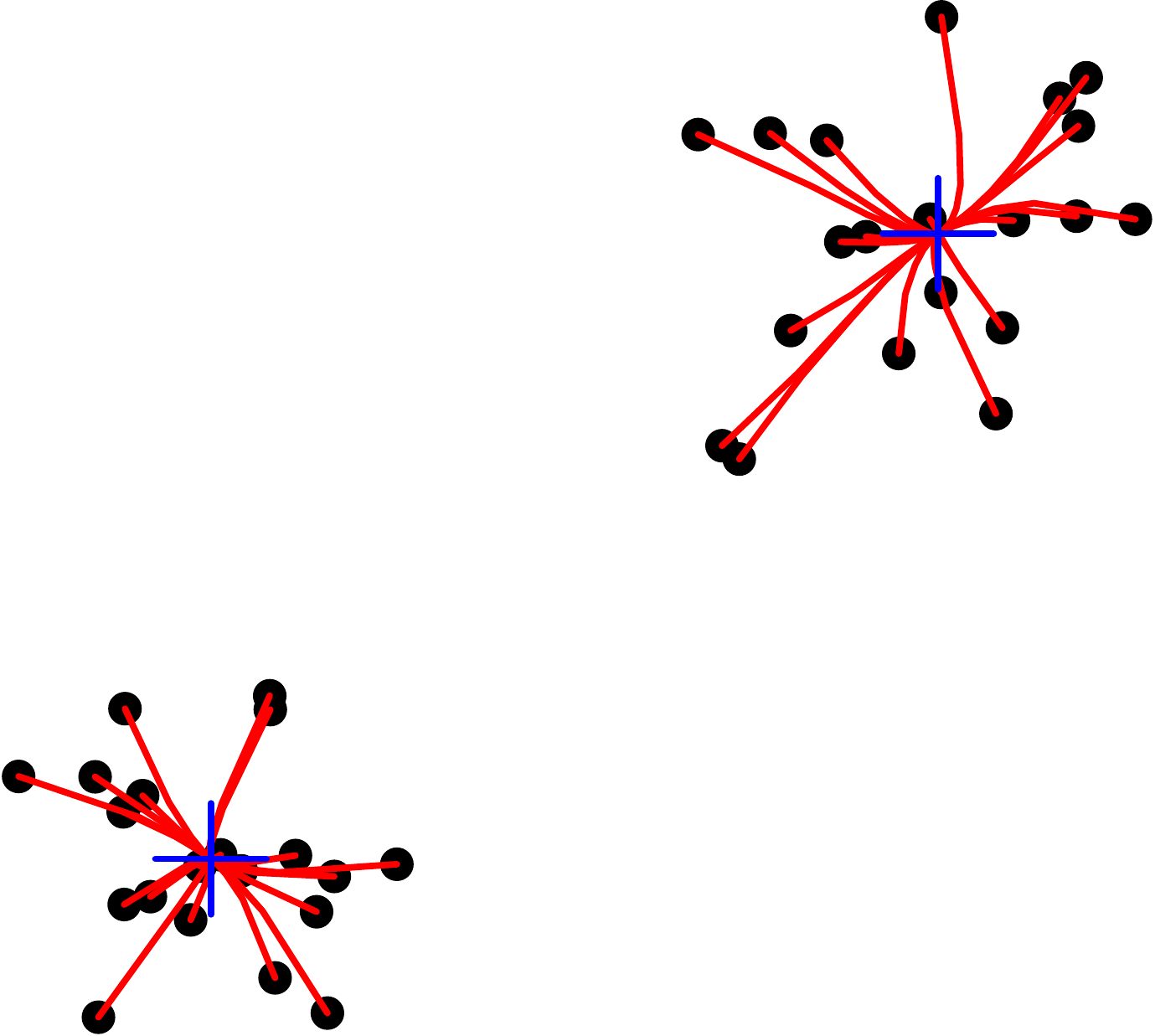}
\end{center}
\begin{center}
\captionof{figure}{\em The mean shift algorithm.
The data are represented by the black dots.
The modes of the density estimate are the two blue cross. 
The red curves show the mean shift paths; each data point moves 
along its path towards a mode
as we iterate the algorithm.}
\label{fig::simplemeanshift}
\end{center}

\section{Density Mode Clustering}
\label{section::clustering}

In this section,
we review mode-based clustering
beginning with mean-shift clustering and
then we move on to the
approach in \cite{rodriguez2014clustering}.

\subsection{Mean-Shift Clustering}

The most common mode-based
clustering method
is mean-shift clustering
\citep{chacon2015population, chacon2013data, chacon,
li2007nonparametric, comaniciu2002mean, arias2015estimation, 
cheng1995mean, genovese2016non}. A simple illustration is in 
Figure \ref{fig::simplemeanshift}. 
The idea is to find modes of the density and then define
clusters as the basins of attraction of the modes.

Let $x$ be an arbitrary point.
If we follow the steepest gradient ascent path starting at $x$,
we will eventually end up at one of the modes.
More precisely,
the gradient ascent path
(or integral curve)
starting at $x$
is the function
$\pi_x: \mathbb{R}\to\mathbb{R}^d$
defined by the
differential equation
\begin{equation}\label{eq::flow}
\pi_x'(t) = \nabla p(\pi_x(t)),\ \ \ \pi_x(0) =x.
\end{equation}

\vspace*{-.6in}
\begin{center}
\begin{tabular}{ll}
\includegraphics[scale=.4]{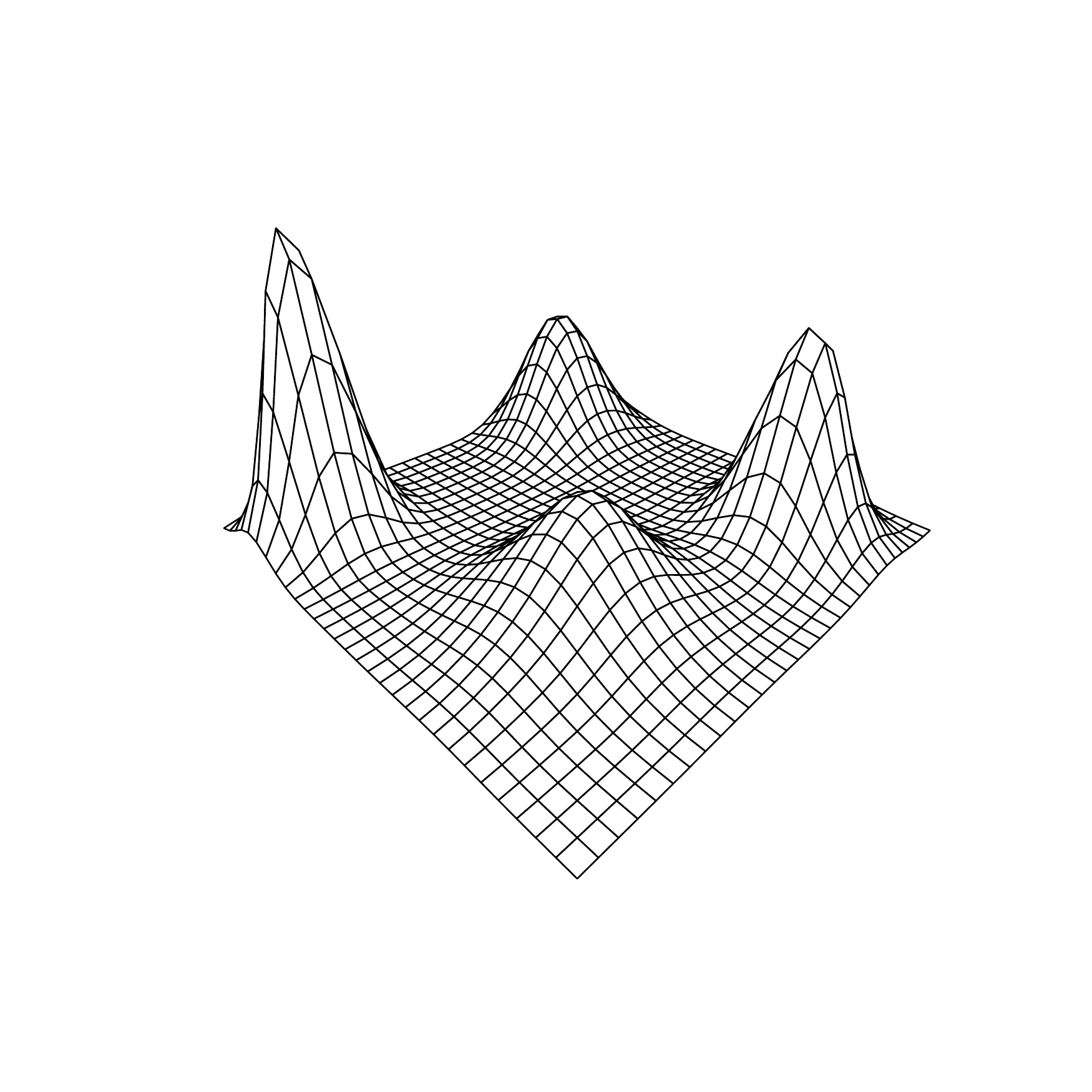}
\includegraphics[scale=.3]{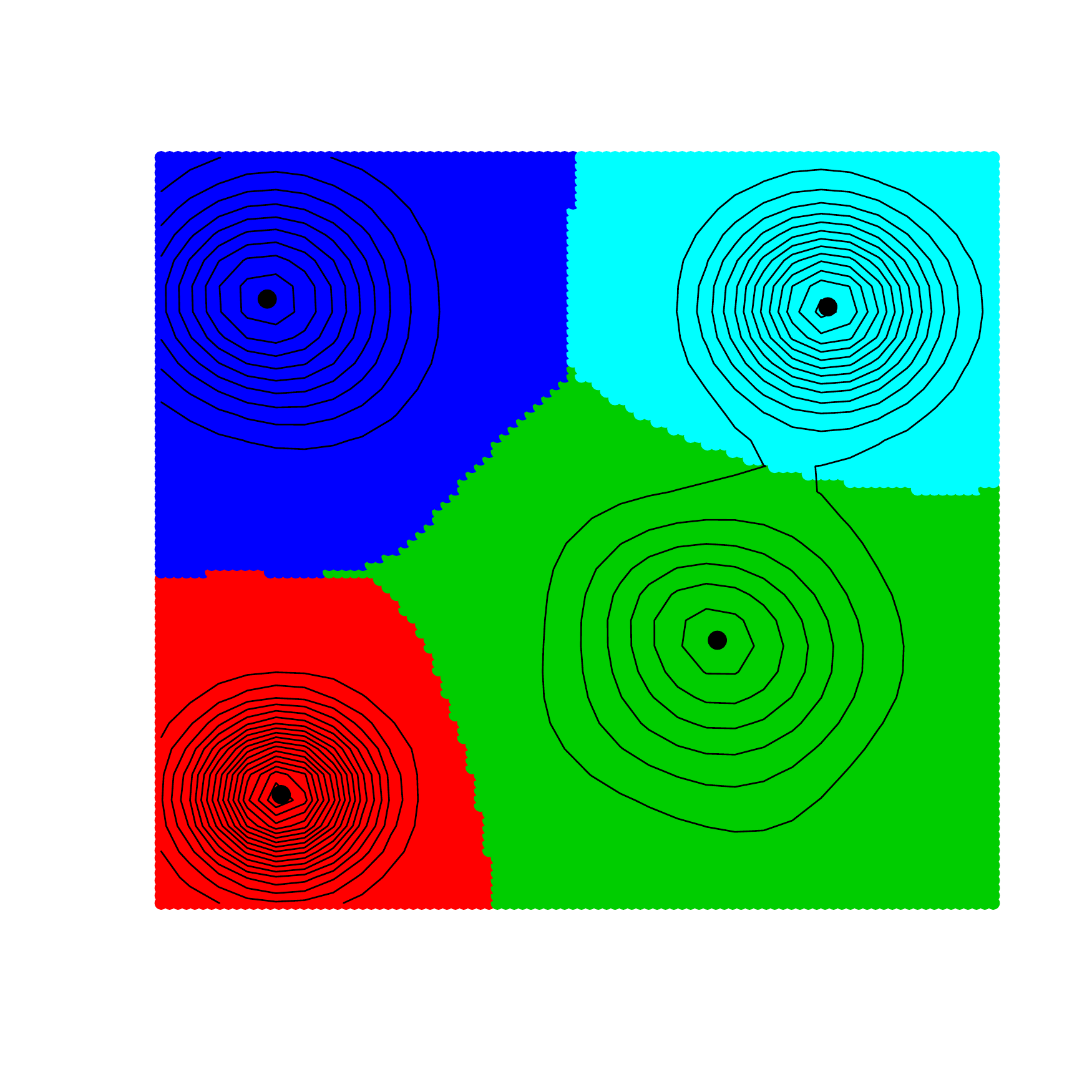}
\end{tabular}
\end{center}
\begin{center}
\captionof{figure}{\em Left: a density with four modes. 
Right: the partition (basins of attraction) of the space 
induced by the modes. These are the population clusters.}
\label{fig::mode-fig}
\end{center}

The {\em destination} of $x$ is defined by
\begin{equation}
{\rm dest}(x) = \lim_{t\to\infty}\pi_x(t).
\end{equation}
It can be shown that, for almost all $x$,
${\rm dest}(x)\in {\cal M}$.
(The exceptions, which have measure 0, lead to saddle points.)
The path $\pi_x$ defines the gradient flow from a point $x$ to 
its corresponding mode.

The {\em basin of attraction} of the mode $m_j$ is the set
\begin{equation}
{\cal C}_j = \Bigl\{x:\ {\rm dest}(x) = m_j\Bigr\}.
\end{equation}
In the mean-shift approach to clustering,
the population clusters are defined to be the basins of attraction
$C_1,\ldots, C_k$.
The left plot in Figure \ref{fig::mode-fig} shows a 
bivariate density with four modes.
The right plot shows the partition (basins of attractions) 
induced by the modes.

To estimate the
clusters, we find the modes
$\hat {\cal M} =\{\hat m_1,\ldots, \hat m_r\}$
of a density estimate $\hat p$.
A simple iterative algorithm called the
{\em mean shift algorithm}
\citep{cheng1995mean, comaniciu2002mean}
can be used to find the modes and to
find the destination of any point $x$
when $\hat p$ is the kernel density estimator:
\begin{equation}\label{eq::kernel}
\hat p(x) = \frac{1}{n}\sum_{i=1}^n \frac{1}{h^d}K\left(\frac{x-X_i}{h}\right)
\end{equation}
with kernel $K$ and bandwidth $h$.
For any given $x$, we define the iteration,
$x^{(0)} \equiv x$,
$$
x^{(j+1)} = \frac{\sum_i X_i K\left(\frac{||x^{(j)}-X_i||}{h}\right)}
{\sum_i K\left(\frac{||x^{(j)}-X_i||}{h}\right)}.
$$
See Figure \ref{fig::simplemeanshift}.
It can be shown that this algorithm is an adaptive 
gradient ascent method,
that approximates the gradient flow defined by
(\ref{eq::flow}).
The convergence of this algorithm is studied in
\cite{arias2015estimation}.

\subsection{The Mode Diagram}

Following \cite{rodriguez2014clustering},
we define
\begin{equation}
\delta(X_i) = 
\min\Bigl\{ ||X_j - X_i||:\ p(X_j) > p(X_i)\Bigr\}.
\end{equation}
That is,
$\delta(X_i)$ is the distance of $X_i$ to the closest point
 with higher 
density. In the case where there are no points with $p(X_j) > p(X_i)$
(in other words, $p(X_i) = \max_j \{ p(X_j):\ j=1,\ldots,n\}$)
we define
$\delta(X_i)=L$
where $L$ is any, arbitrary positive constant.
The choice of $L$ does not matter.
In practice,
\cite{rodriguez2014clustering}
suggest
setting
$L = \max_{i,j}||X_j - X_i||$
which is the diameter of the dataset.
In our examples, this is what we shall do.
For developing theory, it will be convenient to just keep $L$
as any arbitrary, positive constant.

Next we form the {\em mode diagram},
where we plot the pairs
$(p(X_i),\delta(X_i))$.
The intuition is that $\delta(X_i)$ will be small for most points.
But if $X_i$ is close to a local mode, then
$\delta(X_i)$ will be large since the nearest point with a higher density will be at 
another mode.

Hence, modes will show up as isolated points
in the top right of the mode diagram.
A simple example is shown in Figure \ref{fig::FIG3}.
Formally, the {\em mode diagram} is the collection of pairs
\begin{equation}
{\cal D} = \Bigl\{ (p(X_i),\ \delta(X_i)):\ i=1,\ldots, n \Bigr\}.
\end{equation}
The modes can be identified by inspection of the diagram.
In this paper, we suggest a method to separate modes from 
non-modes using linear regression.

In practice, 
we need to estimate $p$.
We will use
the kernel density estimator
defined in (\ref{eq::kernel}).
Then we define
$\hat\delta(X_i) = \min\{ ||X_j-X_i||:\ \hat p(X_j) > \hat p(X_i)\}$.
We have to decide which points on the diagram correspond to modes.
For this purpose,
let $t_n: \mathbb{R}\to\mathbb{R}$ be a given function.
The points $X_i$ such that
$$
\hat\delta(X_i) > t_n(\hat p(X_i))
$$
are the estimated modes.
Denote these points by
$\hat{\cal M} = \{ \hat m_1,\ldots, \hat m_\ell\}$.
We call $t_n$ the {\em threshold function}.

The main contribution of this paper is to study the properties
of the mode diagram.
We shall see that
the mean of $\log \delta(X_i)$ as a function of $\log p(X_i)$ is approximately linear, for non-modes.
But when $X_i$ is close to a mode, $\log \delta(X_i)$ lies far above the line.
This suggests the following method for separating modes from non-modes.
We perform a robust linear regression of
$\log \hat\delta(X_i)$ on
$\log \hat p(X_i)$.
That is, we find $\hat\beta_0$ and $\hat\beta_1$ such that
$$
\log \hat \delta (X_i) \approx \hat\beta_0 + \hat\beta_1 \,  \log \hat p (X_i).
$$

Then we look for large positive outliers.
These are points for which
$\log \hat \delta(X_i) > \hat\beta_0 + \hat\beta_1  \log \hat p(X_i) + M s$
where $s$ is the estimated residual standard deviation 
and $M$ is some large constant; 
we use $M=3$ for the examples in this paper.
These points correspond to modes.
This corresponds to taking the threshold function
\begin{equation}\label{eqn::tn}
t_n(u) = e^{\hat\beta_0+ Ms} u^{\hat\beta_1}.
\end{equation}
The reason for this choice of threshold function
arises from the theory in Section \ref{section::oracle}.

To assign points to modes,
\cite{rodriguez2014clustering} suggest
an approach that avoids the iterations of the mean-shift method.
Instead, 
we assign each point to nearest neighbor with higher density.
This leads each sample point to a mode
without having to recompute the density estimator at any other points.
This is essentially a sample-based approximation to the gradient.
The steps of the algorithm are summarized in Figure \ref{fig::mode-alg}.

This method has several advantages over mean-shift clustering.
We never need to estimate or approximate the gradient of the density.
There is no need for any iterative calculation of the density.
This makes the method fast. However, our focus is not on the 
algorithm but on the mode diagram which gives a nice, 
two-dimensional summary of the clustering information.

\begin{figure}
\fbox{\parbox{5in}{
Input: Data $\{X_1,\ldots, X_n\}$, threshold function $t_n$.
\begin{enum}
\item Compute density estimator $\hat p(X_i)$ at each point.
\item Define $\hat\delta(X_i) = \min \{ ||X_j-X_i||:\ \hat p(X_j) 
              > \hat p(X_i)\}$.
Take $\hat\delta(X_i) = \max_{i,j} ||X_i - X_j||$ if
$\hat p(X_i)> \hat p(X_j)$ for all $j$.
\item Let $\hat{\cal M} = \{X_i :\ \hat\delta(X_i) > t_n(\hat p(X_i))\}$.
\item Cluster assignment: for each $X_i$, move to the closest point 
with higher density.
Continue until a mode is reached.
\item Let $\overline{C}_j$ be all points assigned to 
$m_j \in \hat{\cal M}$.
\end{enum}

Return: $\overline{C}_1,\ldots, \overline{C}_\ell$.
}}
\captionof{figure}{\em The Rodriguez-Laio Algorithm.}
\label{fig::mode-alg}
\end{figure}

\section{The Oracle Diagram}
\label{section::oracle}

In this section we assume that the density $p$ is known.
We then call
${\cal D} = \{ (p(X_i),\delta(X_i)):\ i=1,\ldots, n\}$
the {\em oracle diagram}.
Note that ${\cal D}$ is a point process on $\mathbb{R}^2$.
The variables $\delta(X_i)$ are not independent
since $\delta(X_i)$ depends on the configuration
of the other points.

We need the following definition from Cuevas et al (1990).
A set $S$ is $(\gamma,\tau)$-standard if there exist
$\epsilon_0 >0$ and $\tau\in (0,1)$
such that:
for all $0 < \epsilon \leq \epsilon_0$ and all $x\in S$,
\begin{equation}
\mu\Bigl(B(x,\epsilon)\bigcap S\Bigr)\geq \tau \,\mu(B(x,\epsilon)).
\end{equation}
A set that is standard does not have sharp protrusions.
Our proofs require the assumption that the level sets $\{p>t\}$
are standard.
However, this requires some care.
Suppose that $x=m_j$ where $m_j$ is a mode of $p$.
Let $S = \{y:\ p(y) \geq p(x)\}$.
Then $S \bigcap B(x,\epsilon) = \{m_j\}$ and so
$\mu(S \bigcap B(x,\epsilon)) =0$
and standardness thus fails for points that are modes.
More generally, we cannot lower bound
$\mu(\{ p(y) \geq p(x)\} \bigcap B(x,\epsilon))$
unless $x$ is at least $\epsilon$ far from the modes.
We use the following restricted standardness assumption.
Let $L_x = \{y:\ p(y) > p(x)\}$.

(A4) There exists $\epsilon_0>0$ and $\tau\in (0,1)$ such that,
whenever $\min_j ||x-m_j|| > t$ with $0 < t < \epsilon_0$, we have that
\begin{equation}
\mu\Bigl(L_x \bigcap B(x,t)\Bigr)\geq \tau \mu(B(x,t)).
\end{equation}

In the rest of the section we 
assume that (A1)-(A4) hold.

Before proceeding,
we need a bit more notation.
Recall that
${\cal M} = \{ m_1,\ldots, m_k\}$ is the set of true modes.
Let $m_j\in {\cal M}$.
Let $H(m_j)$ be the Hessian at $m_j$.
Let $J(m_j) = - H(m_j)$ and $\lambda_j$ be the 
smallest eigenvalue of $J(m_j)$.
Note that $\lambda_j >0$.
Since the Hessian is a continuous function,
there exists $\omega_j>0$ such that 
$\lambda_{\rm min}(J(x)) \geq \lambda_j/2$,
for all $x\in B(m_j,\omega_j)$.
Let $\Lambda_j = \sup_{x\in B(m_j,\omega)}\lambda_{\rm max}(J(x))$.

Define 
\begin{equation}\label{eq::tn-def}
\epsilon_n = \left(\frac{r\log n}{n}\right)^{\frac{1}{d}},\qquad
t_n(u) = \left(\frac{C\log n}{n\,u}\right)^{\frac{1}{d}}
\end{equation}
where
\begin{equation}\label{eq::C}
C \geq G^d 2^{d/2} r\max_j p(m_j)
        \left(\frac{\Lambda_j}{\lambda_j}\right)^{d/2}, 
\end{equation}
and
\begin{equation}\label{eq::G}
G \geq \max\Biggl\{\max_j
       \left[3\left(\frac{\lambda_j}{\Lambda_j}\right)^d 
       \frac{1}{2^{d/2}v_d a \tau}\right]^{\frac{1}{d}},\,1\Biggr\} 
\end{equation}
where $v_d$ denotes the volume of the unit ball in $\mathbb{R}^d$
and $r> 1/(av_d)$.

{\bf Remark:}
{\em The constants --- such as $C$, $r$, $G$ and so on --- are only used to state
the theoretical results.
The actual procedure described in Section 6 does not require these constants.}

Because $p$ is Morse,
the modes are isolated points.
It follows that
there exists some $c>0$ such that
$B(m_s,c\omega_s)\bigcap B(m_t,c\omega_t)=\emptyset$ for
all $1\leq s < t \leq k$.
Without loss of generality, 
we assume that $c=1$.
Hence,
$B(m_s,\omega_s)\bigcap B(m_t,\omega_t)=\emptyset$ for
$1\leq s < t \leq k$.

We assume that (A1)-(A4) hold
in the rest of the paper.

\subsection{Properties of the Mode Diagram}
\label{section::mode-diagram}

We first need to define $k$ sample points that can be considered 
to be sample modes.
These are the points that will be in the upper right 
portion of the mode diagram. (For a unimodal density, 
this would just be the point $X_i$ that maximizes $p(X_i)$.)
We define the sample point $X_j \in B(m_j,\omega_j)$ to be a 
{\em sample mode} if 
$p(X_j) \geq p(X_i)\ \ {\rm for\ all\ } X_i \in B(m_j,\omega_j)$.
We renumber the points so that
$X_1,\ldots, X_k$ denote the $k$ sample modes.
Note that these sample modes are not known since they depend on $m_j$
and $\omega_j$.
But, as we shall see, we can identify them by using the mode diagram.
The next result shows 
that $X_j$ is close to $m_j$ and
that $\delta(X_j)$ is bounded below by a constant.

\begin{theorem}
\label{thm::sample-mode}
Let
\begin{equation}
\psi^2_{n,j} = \frac{2\,G^2\,\Lambda_j}{\lambda_j}\;\epsilon_n^2
\end{equation}
where 
$G$ was defined in (\ref{eq::C}).
Also, let
$\psi_n = \max_j \psi_{n,j}$.
If $X_j$ is a sample mode then:

(i) $X_j \in B(m_j,\psi_{n,j})$.

(ii) Any sample point $X_i$ such that $p(X_i) > p(X_j)$ is 
far from $X_j$; specifically $\delta(X_i) \geq \omega_j/2$.
\end{theorem}

\bigskip

Now let $m_j, j=1,\ldots, k$, denote the modes of $p(x)$, and 
let $X_j, j=1,\ldots, k$, be the local modes.
Define $\Gamma = \bigcup_{j=1}^k B(\,m_j,\psi_{n})$
and divide the dataset into three groups:
\begin{equation}
\label{eq::sample-modes}
{\cal X}_1 = \biggl \{X_1,\ldots, X_k \biggl\}, \ \ \ 
{\cal X}_2 = \biggl\{X_i:\ X_i\in \Gamma, X_i\notin {\cal X}_1\biggl\},\ \ \ 
{\cal X}_3 = \biggl\{ X_i \in \Gamma^c \biggl\}.
\end{equation}
Note that ${\cal X}_1$ is precisely the set of sample modes.
Theorem \ref{thm::oracle-theorem}, below, shows that
$\delta(X_j)$ is bounded away from 0 for the points in ${\cal X}_1$
(and hence they lie above the threshold function)
while $\delta(X_i)$ lies below the threshold functions for all
$X_i$ in ${\cal X}_2$ and $X_i$ in ${\cal X}_3$.

{\bf Remark:}
{\em If the assumption that $p(x)\geq a >0$ is dropped,
then ${\cal X}_3$ needs to be re-defined as
${\cal X}_3 = \bigl\{ X_i \in \Gamma^c \bigr\}\bigcap 
\bigl\{X_i:p(X_i)\geq  n^{-1/(d+2)} \bigr\}$.}

\begin{theorem} \label{thm::oracle-theorem}

(i) Let $X_j \in {\cal X}_1$ be the sample mode in $B(m_j,\psi_{n,j})$. 
Then $p(X_j) = p(m_j) + O_P(\epsilon_n^2)$ and
$\delta(X_j)/t_n(p(X_j))\to \infty$.

(ii) For all $X_i \in {\cal X}_2$, we have
$\delta(X_i) \leq  t_n(p(X_i))$.

(iii) $P^n\Bigl(\delta(X_i) \leq t_n(p(X_i))\ 
{\rm for\ all\ }X_i \in {\cal X}_3\Bigr)\to 1$.
\end{theorem}

\subsection{The Limiting Distribution}

To get more information about the shape of the mode diagram
we show that for any $x$ that is not a mode,
the distribution of $n\delta^d(x)$ 
only depends on $p(x)$ and converges to an
exponential random variable with mean $1/(p(x)\,\tau\, v_d)$
where $v_d$ is the volume of the unit ball.
This means that $\delta(x) \approx (n \,\tau\, v_d\, p(x))^{-1} E$
where $E\sim {\rm Exp}(1)$ and that a plot of $\log \delta(X_i)$ versus
$\log p(X_i)$ should look linear
for all $X_i$'s not close to a mode.
On the other hand if $x$ is a mode, then
$n\delta^d(x)\to \infty$.

We will need the following stronger version of (A4).
Recall that
$L_x = \{y:\ p(y) > p(x)\}$.

(A4') There exists $\tau\in (0,1)$ such that,
for any $x\notin {\cal M}$,
\begin{equation}
\lim_{t\to 0}\frac{\mu\Bigl(L_x \bigcap B(x,t)\Bigr)}{\mu(B(x,t))} = \tau.
\end{equation}

\newpage
\begin{theorem}
\label{theorem::exponential}
Suppose that (A1), (A2), (A3) and (A4') hold and that
$x$ is not a mode.
Then the random variable
$n\; \delta^d(x)$ converges in distribution to an
exponential random variable, with parameter $p(x)\tau\;v_d$.
If $x$ is a mode, then $n\delta^d(x)\to \infty$.
\end{theorem}

{\bf Remark:}
{\em We can allow a different limit $\tau(x)$ for each $x$.
In this case, the limiting distribution is 
exponential with parameter $p(x)\tau(x)\;v_d$.}

\subsection{The Linear Heuristic}

The results in the previous sections
show that, for non-modes,
$\log \mathbb{E}[n\,\delta^d(X)]$ should be approximately 
linear in $\log p(x)$.
On the other hand, points closest to modes will lie far above the threshold.
This suggests the following approach:
plot $\log p(X_i)$ versus
$\log \delta(X_i)$.
Most points will fall below some line.
A few points will be above the line.
In Section \ref{section::threshold},
we will fit a robust linear regression to the log-mode diagram.
The outliers above the line will indicate the modes.
We pursue this idea in Section \ref{section::threshold}.

\section{The Estimated Mode Diagram}
\label{section::not-oracle}

Since $p$ is not known, we have to estimate the diagram.
Let $\hat p$ denote the kernel density estimator
and let
\begin{equation} \label{eqn::deltahat}
\hat\delta(X_i) = 
\min\Bigl\{ ||X_j - X_i||:\ \hat p(X_j) > \hat p(X_i)\Bigr\}.
\end{equation}
As before,
if there are no points with $\hat p(X_j) > \hat p(X_i)$
we define
$\delta(X_i)=L$
where $L$ is any positive constant.
The estimated diagram is 
\begin{equation}\label{eq::EstMode}
\hat{\cal D} = \Bigl\{ (\hat p(X_i),\hat\delta(X_i)):\ i=1,\ldots, n\Bigr\}.
\end{equation}

{\bf Remark.}
{\em An alternative approach to defining the estimated diagram is as follows.
We draw a bootstrap sample
$X_1^*,\ldots, X_N^*$ from $\hat p$.
We then define
\begin{equation}
\hat\delta(X_i^*) = 
\min\Bigl\{ ||X_j^* - X_i^*||:\ \hat p(X_j^*) > \hat p(X_i^*)\Bigr\}
\end{equation}
and
\begin{equation}
\hat{\cal D}^* = \Bigl\{ (\hat p(X_i^*),\hat\delta(X_i^*)):\ i=1,\ldots, N\Bigr\}.
\end{equation}
This approach has the advantage that we can take $N$ to be much larger
than $n$.
This gives a more accurate summary $\hat p$.
On the other hand, if $n$ is huge, we might even take $N$ smaller than $n$ 
to reduce computation.
At any rate, by sampling from $\hat p$ we have more control.
We will not pursue this bootstrap approach in this paper.}

The rest of the section is devoted to showing that
$\hat{\cal D}$
has the same behavior as the oracle diagram
in Theorem 2.
First, we recall some facts about $\hat p$.
Let ${\cal M} =\{m_1,\ldots, m_k\}$
denote the modes of $p$ and let
${\cal C} =\{c_1,\ldots, c_r\}$ denote the remaining critical points of 
$\hat{p}$. Let $\hat g$ be the gradient of $\hat p$ and
let 

\bigskip
$\hat H$ be the Hessian.
Then, with high probability, for all large $n$,
$\hat p$ is Morse the the same number of critical points as $p$.
This is summarized in the next result.

\begin{lemma}
\label{lemma::kernel}
Assume (A1)-(A3).
Take the bandwidth to be 
$h_n \asymp n^{-1/(d+6)}$.
Let 
\begin{equation}\label{eq::rn-and-sn}
r_n = a_1(\log n/n)^{2/(4+d)},\ \ \ 
s_n = a_2(1/n)^{2/(d+6)}
\end{equation}
where $a_1$ and $a_2$ are positive constants.
There exists a sequence of events ${\cal A}_n$ such that
$P^n({\cal A}_n)\to 1$ and such that, on ${\cal A}_n$:\\
(i) $\sup_{x\in {\cal X}}||\hat p(x) - p(x)|| \leq r_n$.\\
(ii) $\hat p$ is Morse and has exactly $k$ modes
$\hat m_1 , \ldots, \hat m_k$ with
$\max_j||\hat m_j - m_j|| \leq s_n$.\\
(iii) The remaining critical points 
$\hat {\cal C} =\{\hat c_1,\ldots, \hat c_r\}$ of $\hat p$ also satisfy
$\max_j||\hat c_j - c_j|| \leq s_n$.
(iv) $\sup_x||\hat g(x) - g(x)||_\infty= o_P(1)$ and
$\sup_x \max_{j,k}||\hat H_{jk}(x) - H_{jk}(x)||= o_P(1)$,
and the supremum of the third derivative is $o_P(1)$.
\end{lemma}

For proofs of these facts, see \cite{genovese2016non} and
\cite{chazal2017robust}.
In what follows, we assume that the event ${\cal A}_n$ holds.
In particular, $\hat p$ is Morse.
In what follows,
we refer to positive constants
$c_1, c_2$ which come from
Lemma \ref{lemma::killer}
in the appendix.

As in the previous section,
we may find constants $\omega_j$ such that
the balls $B_j=B(\hat m_j,\omega_j)$ are disjoint
and each contains at least one data point.
For $j=1,\ldots, k$ let
$X_j = \argmin_{X_i\in B_j}\hat p(X_i)$.
Let ${\cal X}= \{X_1,\ldots,X_n\}$.
Define
$$
{\cal X}_1 = \{X_1,\ldots, X_k\},\ \ 
{\cal X}_2 = \Biggl\{ X_i:\ X_i \in \bigcup_{j=1}^k B(\hat m_j,c_1 \epsilon_n)\Biggr\},\ \ 
{\cal X}_3 = {\cal X}- ( {\cal X}_1 \bigcup {\cal X}_2).
$$

As before
define $t_n(\hat p(x)) = (C \log n/(n \hat p(x)))^{1/d}$.
In what follows, we sometimes write
$t_n(x)$ as short for
$t_n(\hat p(x))$.
The behavior of the diagram
in this case is essentially the same as the oracle diagram
as the next result shows.
The proof is much more complicated since
$\hat p$ is a random function and is obviously
correlated with the data.

\begin{theorem}
\label{thm::about-threshold}
Let $t_n$ be defined as in \eqref{eqn::tn}.
Then:

(i) There exists $c>0$ such that, with probability tending one,
$\hat\delta(X_i)\geq c$ for all $X_i \in {\cal X}_1$.
Hence, $\hat\delta(X_i) > t_n(\hat p(X_i))$.

(ii) $\hat\delta(X_i) \leq c_2 t_n(\hat p(X_i))$ for all 
     $X_i\in {\cal X}_2$.

(iii) For ${\cal X}_3$ we have that
$$
P^n\Bigl( \hat\delta(X_i) \leq c_2 t_n(\hat p(X_i))\ 
        {\rm for\ all\ }X_i \in  {\cal X}_3\Bigr)\to 1.
$$
\end{theorem}

\section{Choosing the Threshold Using Robust Regression}
\label{section::threshold}

We now know that both ${\cal D}$ and
$\hat{\cal D}$ have the following behavior.
There are $k$ points $X_1,\ldots, X_k$, corresponding to the $k$ modes,
such that
$\hat p(X_j)$ and $\hat{\delta}(X_j)$ are large.
For the remaining points,
$\hat\delta(X_i)$ is small.
Specifically, $\hat\delta(X_i) < t_n(\hat p(X_i)) = 
(C \log n/(n \hat p(X_i)))^{1/d}$.
In other words,
for non-modes,
the points
$(\log \hat\delta(X_i),\log \hat p(X_i))$
should fall on or below a line.

The modes could be selected visually by examining 
the log-log mode diagram.
Alternatively, if we perform a robust linear regression of
$\log \hat\delta(X_i)$ on $\log \hat p(X_i)$,
we expect the modes to show up as outliers.
Let $\hat\beta_0$ and $\hat\beta_1$ be the estimated intercept 
and slope from the regression.
Thus,
$$
\log \hat \delta(X_i) \approx \hat\beta_0 + \hat\beta_1 \log p (X_i).
$$

Now we look for large positive outliers.
These are points for which
$\log \delta(X_i) > \hat\beta_0 + \hat\beta_1 \log p(X_i) + M s$
where $s$ is the estimated residual standard deviation and 
$M$ is some large constant.
We set $M=3$ in the examples of this paper. 

\smallskip
{\bf Remark.}
{\em It is possible to define some post-processing diagnostics
to make sure that the claimed modes are, in fact, modes.
For example, if $X_j$ is declared a mode,
and ${\cal N}$ is a set of neighbors of $X_j$,
then we can check that
$\hat p(X_j) - \hat p(X_i)>0$
for $X_i \in {\cal N}$.}

\section{Examples}
\label{section::examples}

A first example illustrating the theory in the paper 
was presented in Figure \ref{fig::FIG3}. 
That picture shows a simple two-dimensional data set, with 
two well separated clusters. In Figure \ref{fig::FIG3}
the threshold function $t_n(\hat{p}(X_i))$ was obtained 
as described in Section \ref{section::threshold}.
A histogram of $\hat{\delta}(X_i)$ was also added for completeness.

The examples of this section consist of data with five or four 
clusters in two dimensions, a three-dimensional data-set with
four clusters, and a 15-dimensional data set with two clusters.
The data sets can be enriched with added
random noise. 

Two two-dimensional data sets are in Figure \ref{fig::newdata}. 
The left panel shows five separate clusters, with 
shapes that are parts of a {\em broken circle}. They
consist of 500 points. The data set in the right panel 
present 400 points clustered in the shape of four 
crescents. 
The density $p$, of the {\em broken circle} data, is
estimated with the kernel function $\hat{p}$ in 
\eqref{eq::kernel}. The left panel in Figure 
\ref{fig::ModeCircle} shows the robust regression 
line of $\log(\hat{p})$ versus $\log(\hat{\delta})$.
The plot of residuals is shown in the right panel. This figure
shows five outliers corresponding to 
the five modes in the {\em broken circle}.

\begin{center}
\begin{tabular}{ll}
\includegraphics[scale=.35]{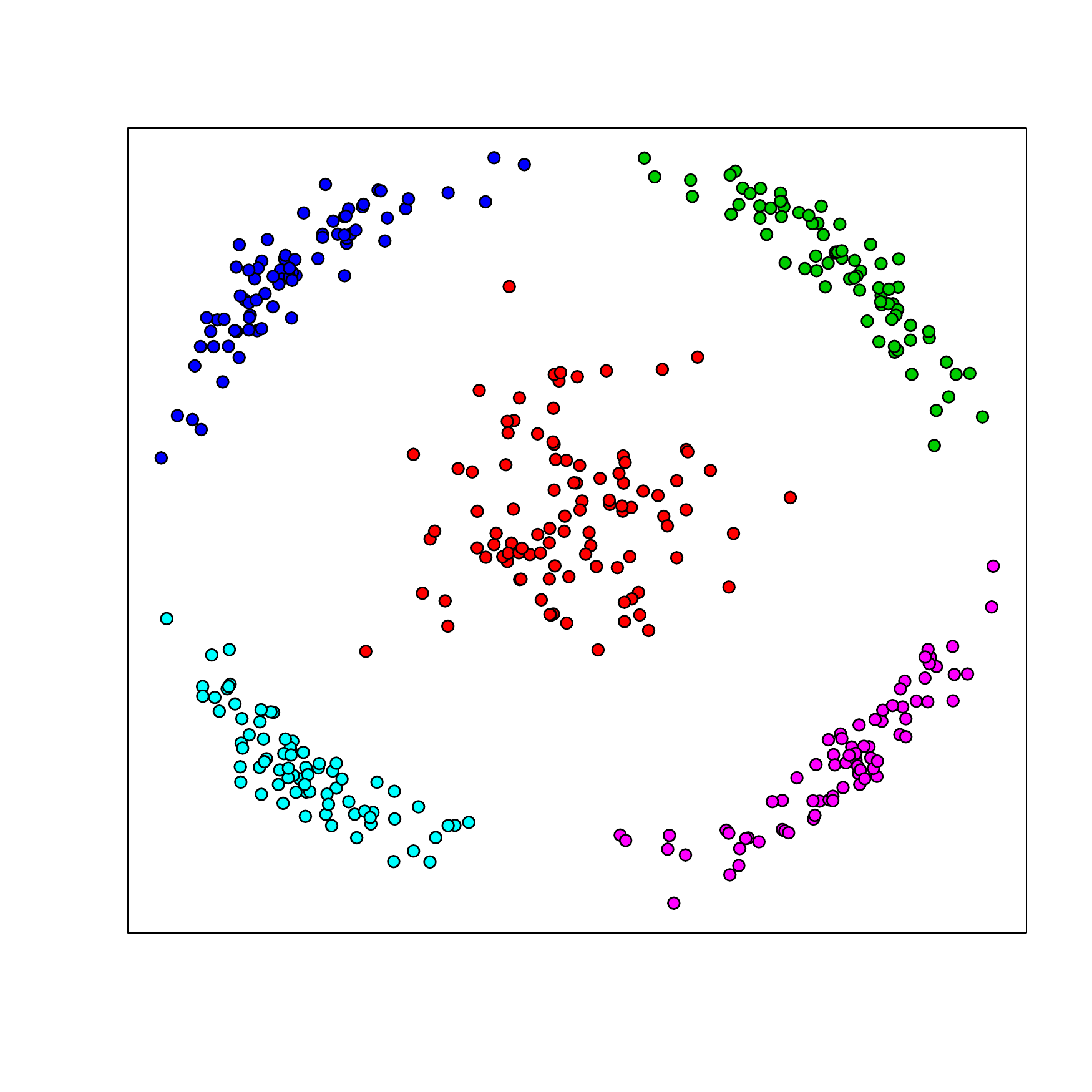} 
\includegraphics[scale=.35]{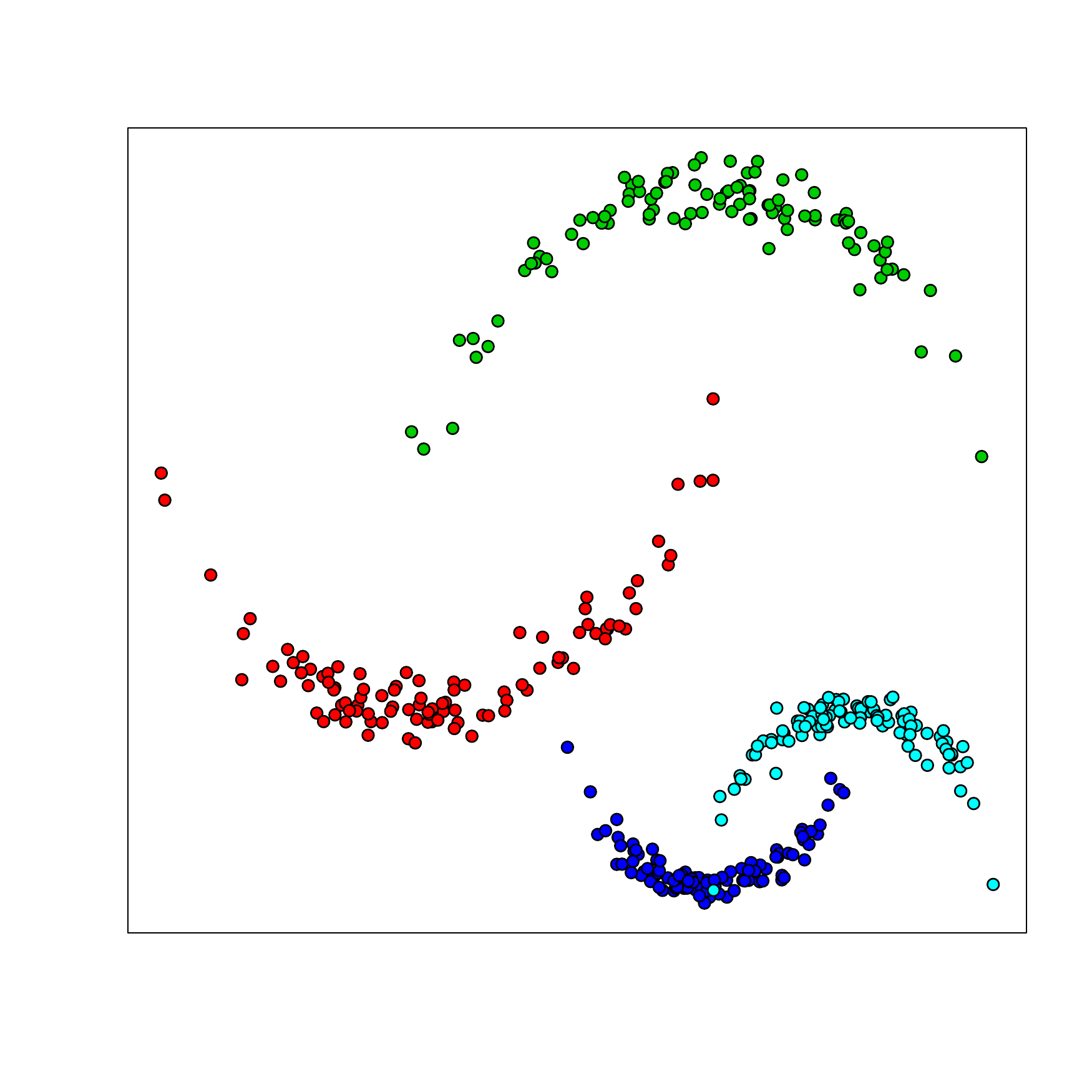} 
\end{tabular}
\end{center}
\begin{center}
\captionof{figure}{\em Left: broken circle data. 
Right: four crescents}
\label{fig::newdata}
\end{center}

\begin{center}
\begin{tabular}{ll}
\includegraphics[scale=.35]{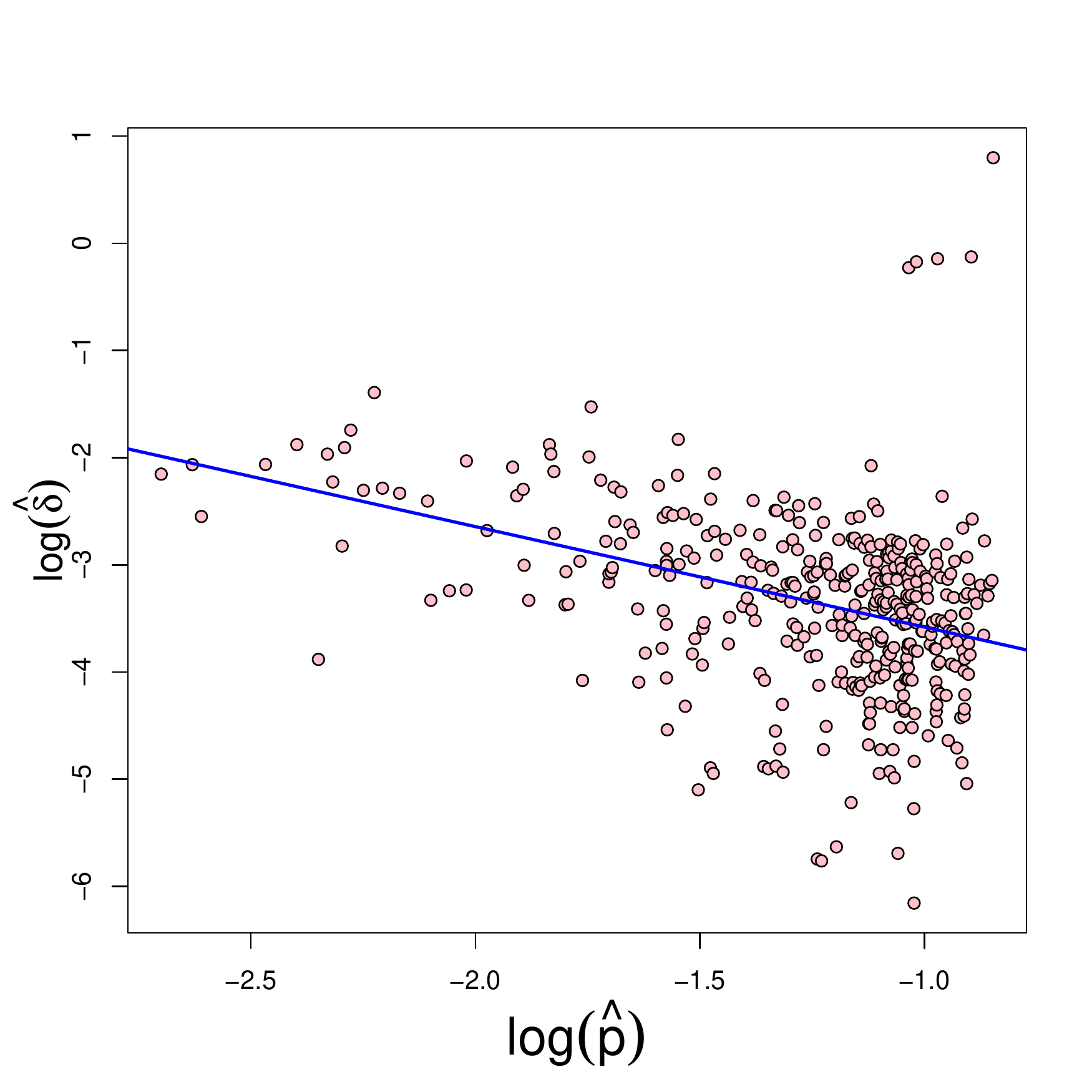} 
\includegraphics[scale=.35]{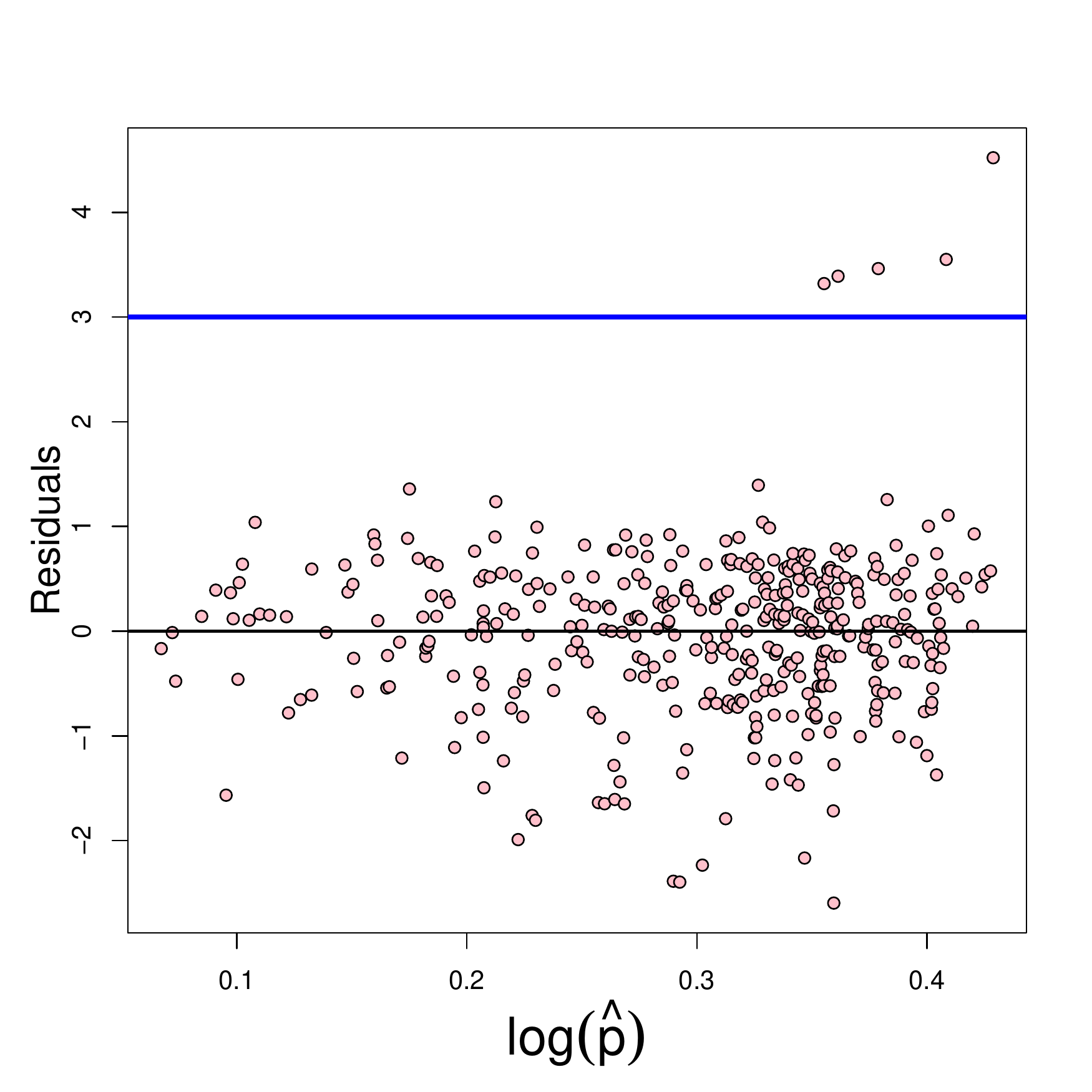} 
\end{tabular}
\end{center}
\begin{center}
\captionof{figure}{\em Broken Circle data. 
Left: Robust regression line for $log(\hat{p})$ vs. 
$\log(\hat{\delta})$}. Right: residuals from robust regression
\label{fig::ModeCircle}
\end{center}

\vspace{-.4in}

\begin{center}
\begin{tabular}{ll}
\includegraphics[scale=.35]{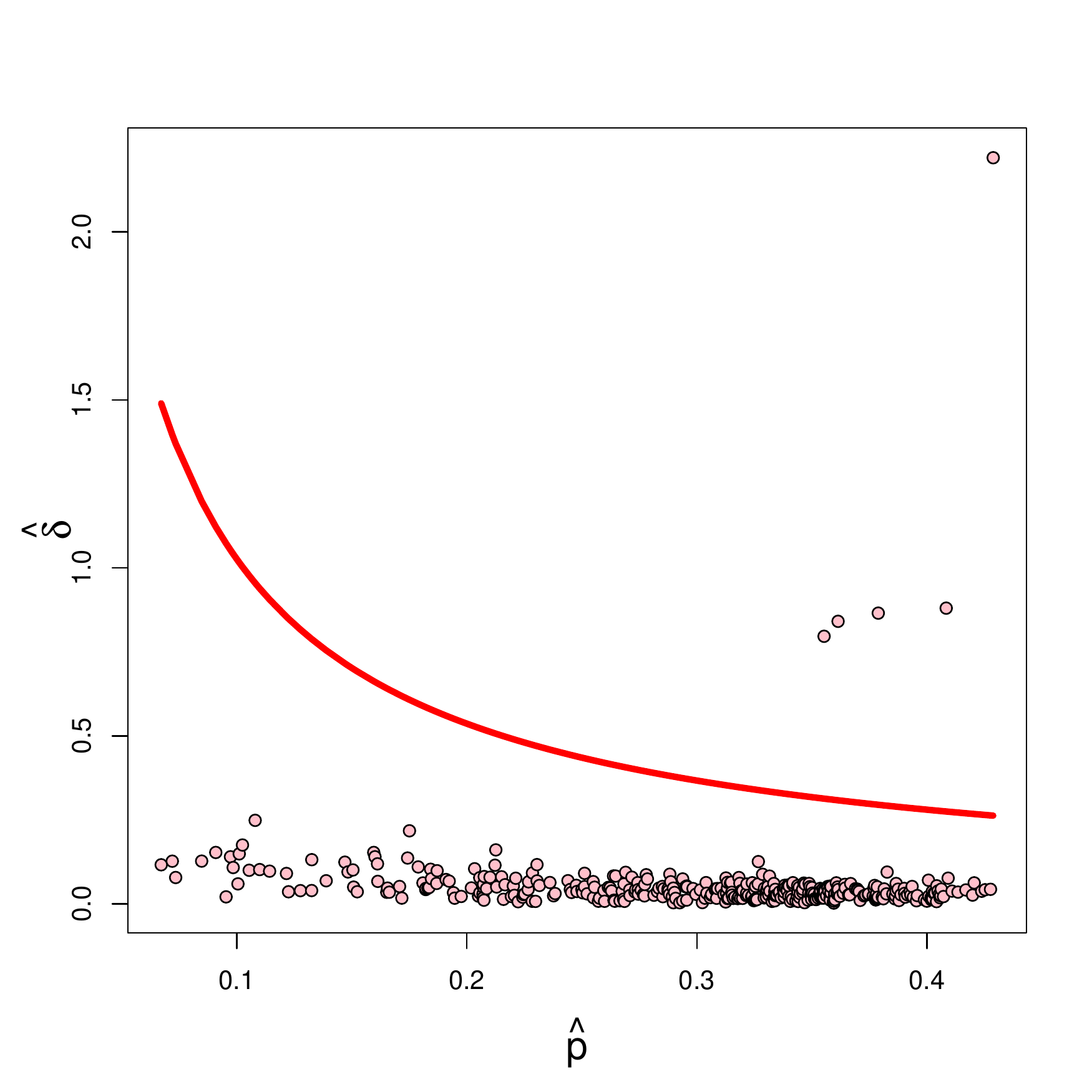} 
\includegraphics[scale=.35]{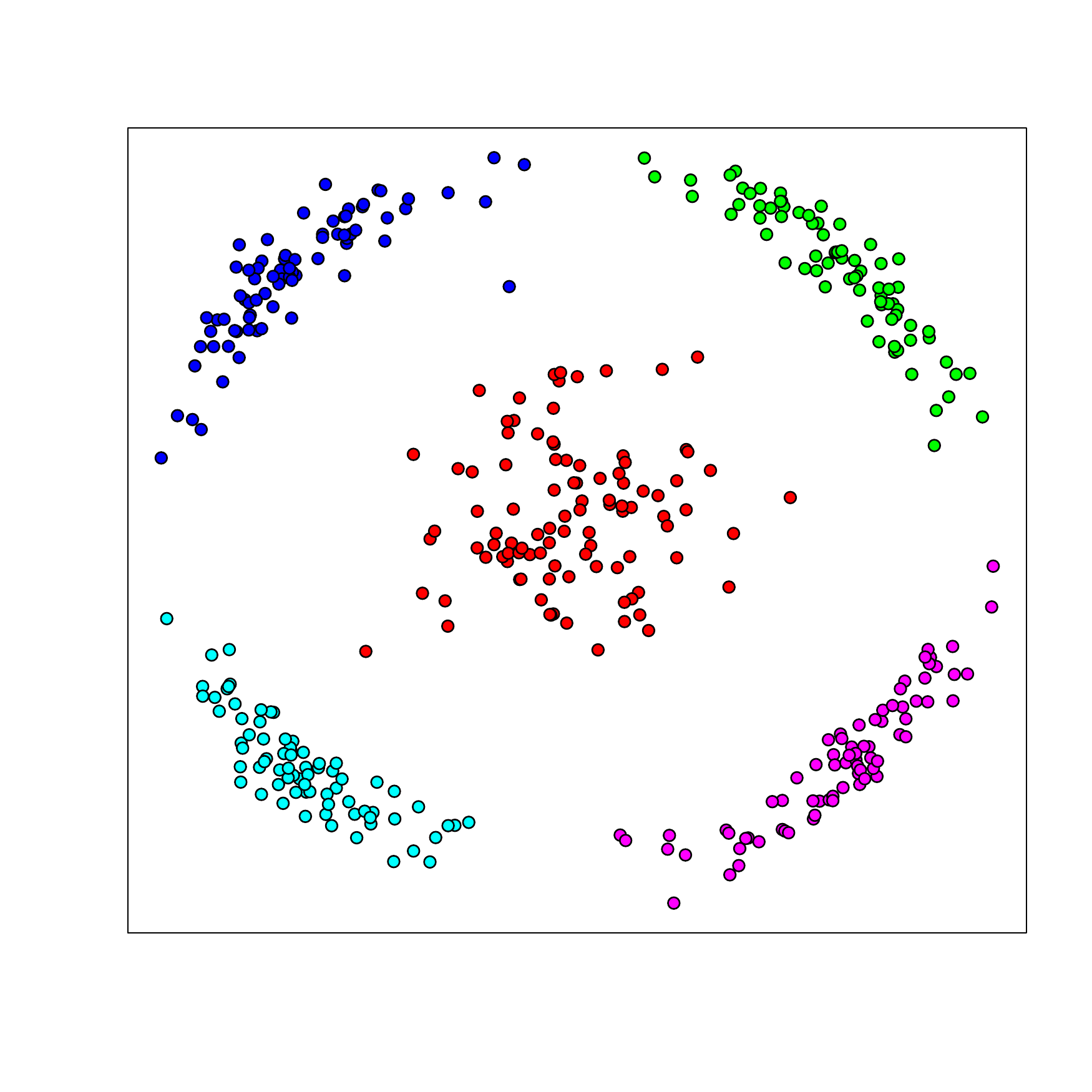} 
\end{tabular}
\end{center}
\begin{center}
\captionof{figure}{\em Broken Circle data. Left: Estimated 
Mode diagram $\cal{\hat{D}}$, with the threshold function.
Right: Estimated clusters.}
\label{fig::ALTRO}
\end{center}

\vspace{-.5cm}

The estimated mode diagram $\cal{\hat{D}}$ from 
\eqref{eqn::deltahat} and \eqref{eq::EstMode} is in
the the left panel of Figure \ref{fig::ALTRO}. It 
displays values of $\hat{p}(X_i)$ and $\hat{\delta}(X_i)$ 
for each data point $X_i$, as well as the threshold function 
$t_n(\hat{p}(X))$, obtained from the robust regression 
line of Figure \ref{fig::ModeCircle}.
The diagram clearly shows the five outliers at the top~corner 
above $t_n(\hat{p}(X))$. They are the
points with large values of $\hat{p}$ and $\hat{\delta}$ or
the modes of the density. Finally, the estimated 
clusters for the {\em broken circle} data, 
in the right panel of Figure \ref{fig::ALTRO}. 

Consider now the dataset in the left panel of 
Figure \ref{fig::CrescNoise}. It
consists of 400 points from the four crescent data of 
Figure \ref{fig::newdata}, augmented with 200 points of 
uniform random noise. Despite the added noise, the mode
diagram $\cal{\hat{D}}$, in the right panel, correctly
identified four modes.

\vspace{-.5in}
\begin{center}
\begin{tabular}{ll}
\includegraphics[scale=.35]{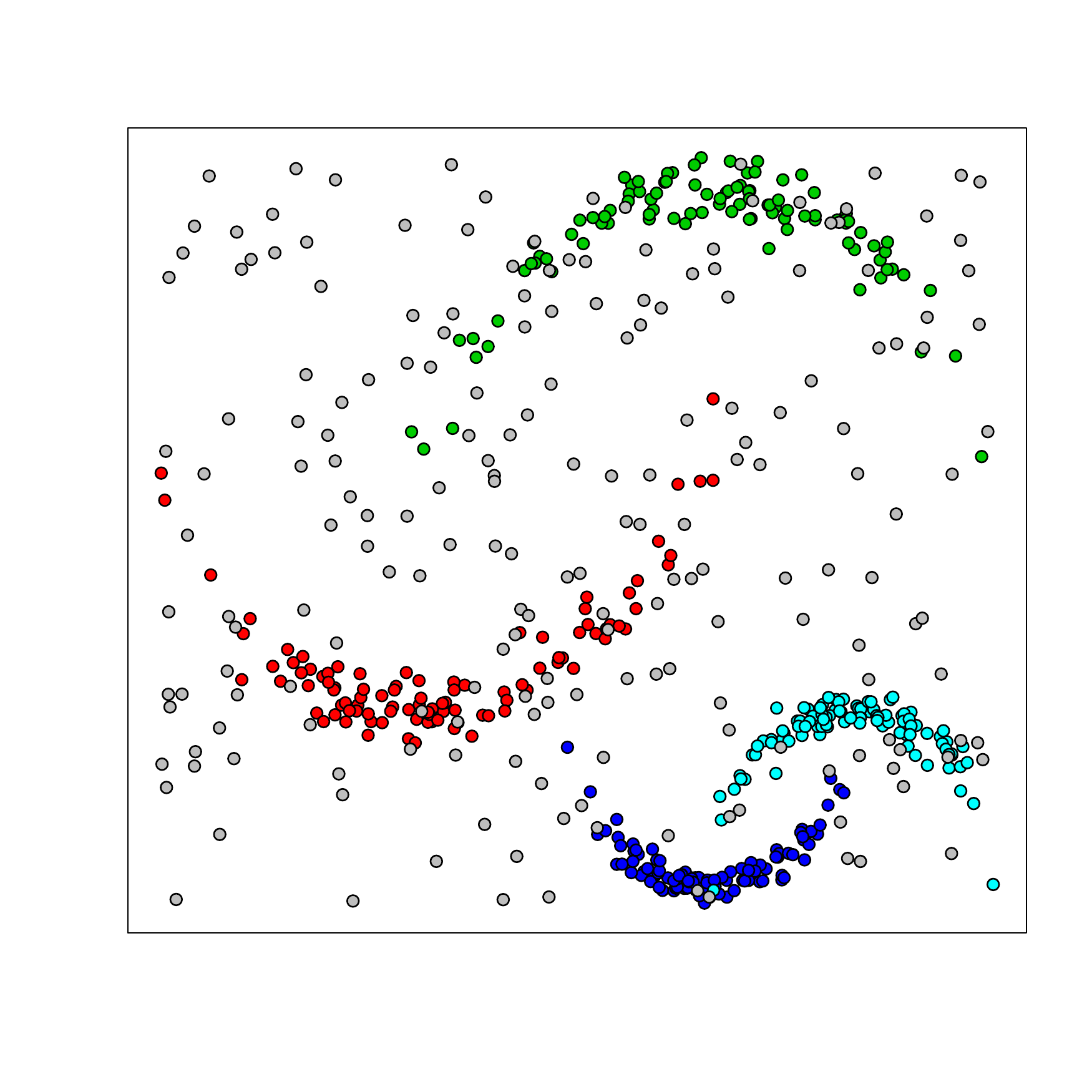} 
\includegraphics[scale=.35]{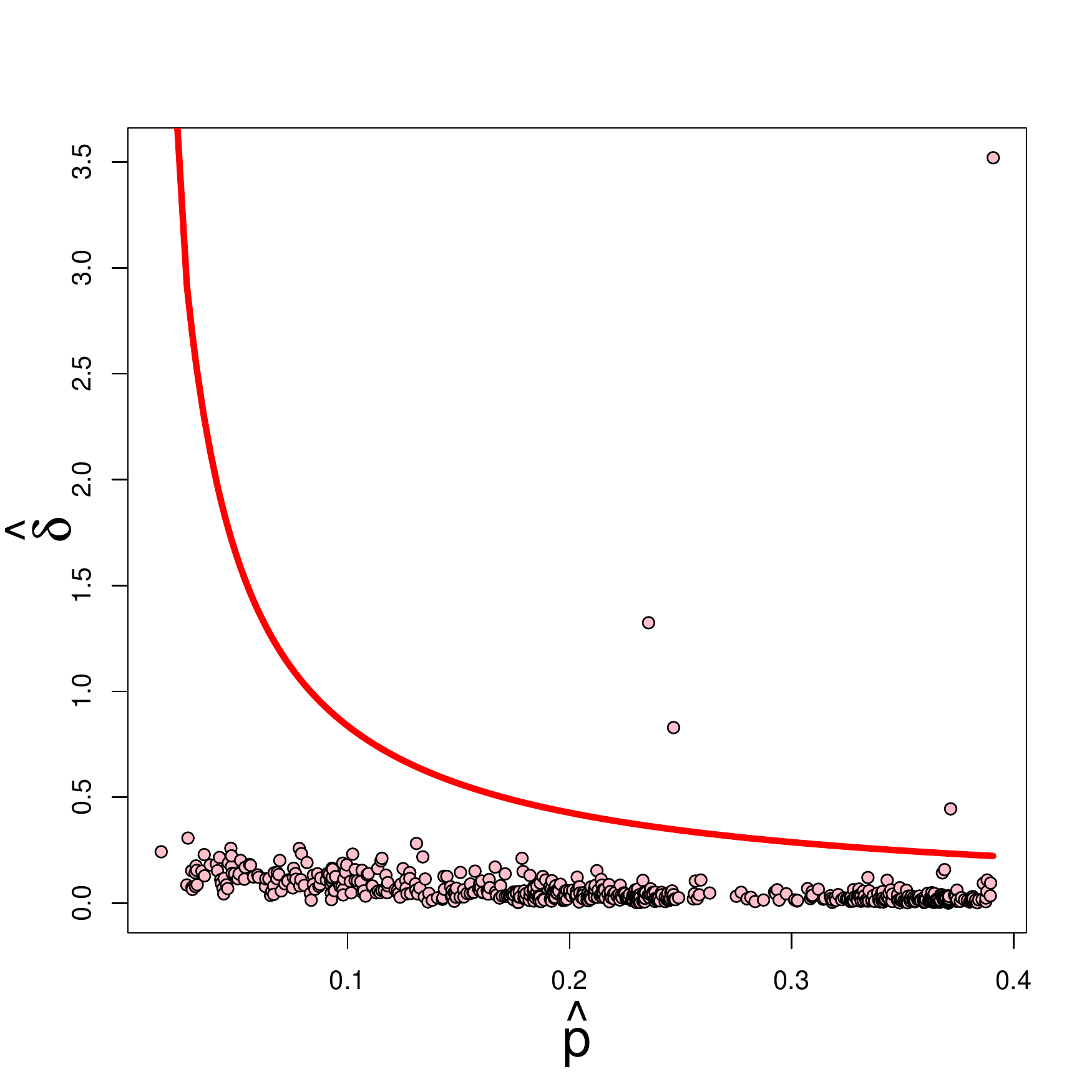} 
\end{tabular}
\end{center}
\vspace{-.4in}
\begin{center}
\captionof{figure}{\em Left: New dataset. Four crescent clusters, 
with added uniform noise. Right: Estimated Mode diagram 
$\cal{\hat{D}}$, threshold function.}
\label{fig::CrescNoise}
\end{center}

\vspace{-.5in}

\begin{center}
\begin{tabular}{ll}
\includegraphics[scale=.35]{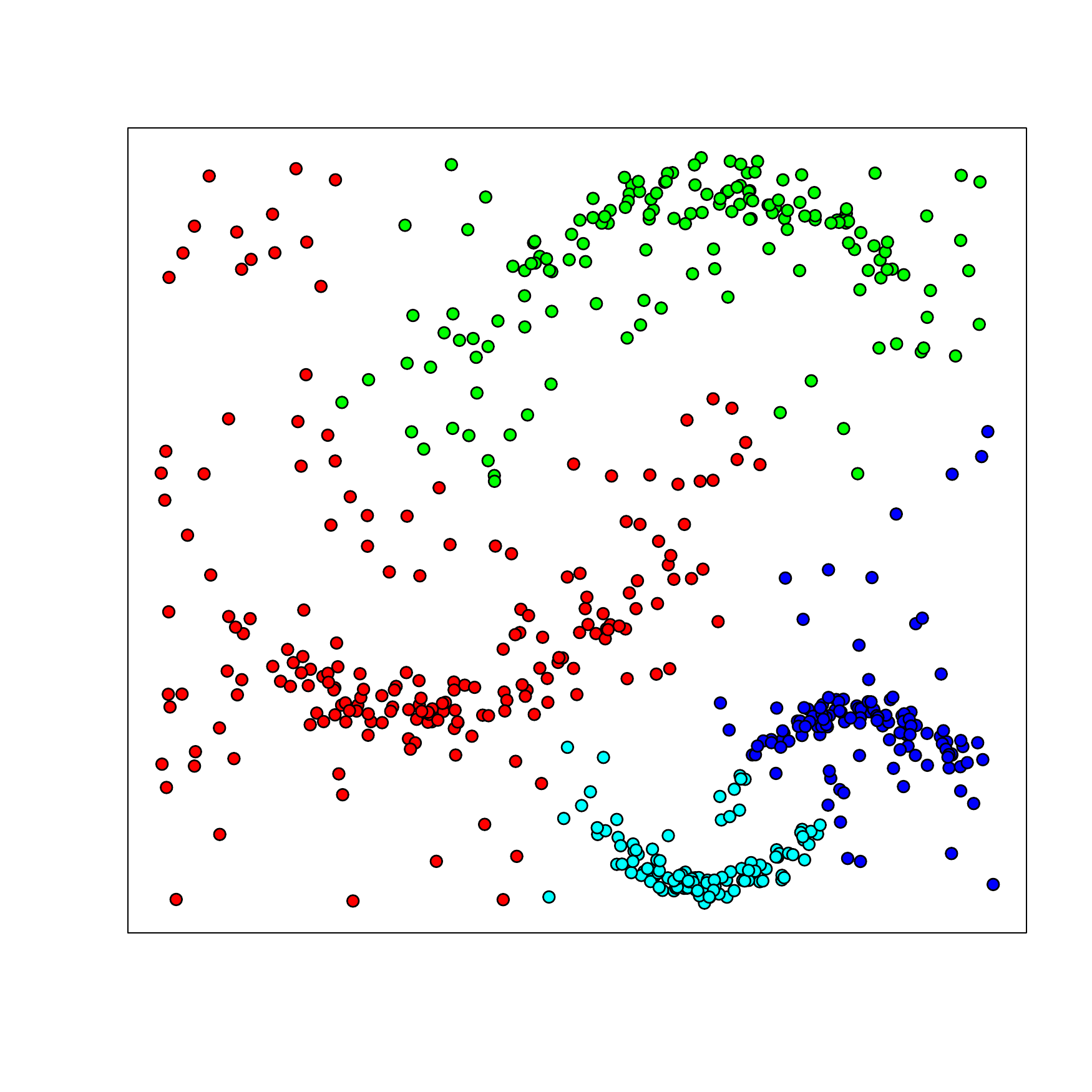} 
\includegraphics[scale=.35]{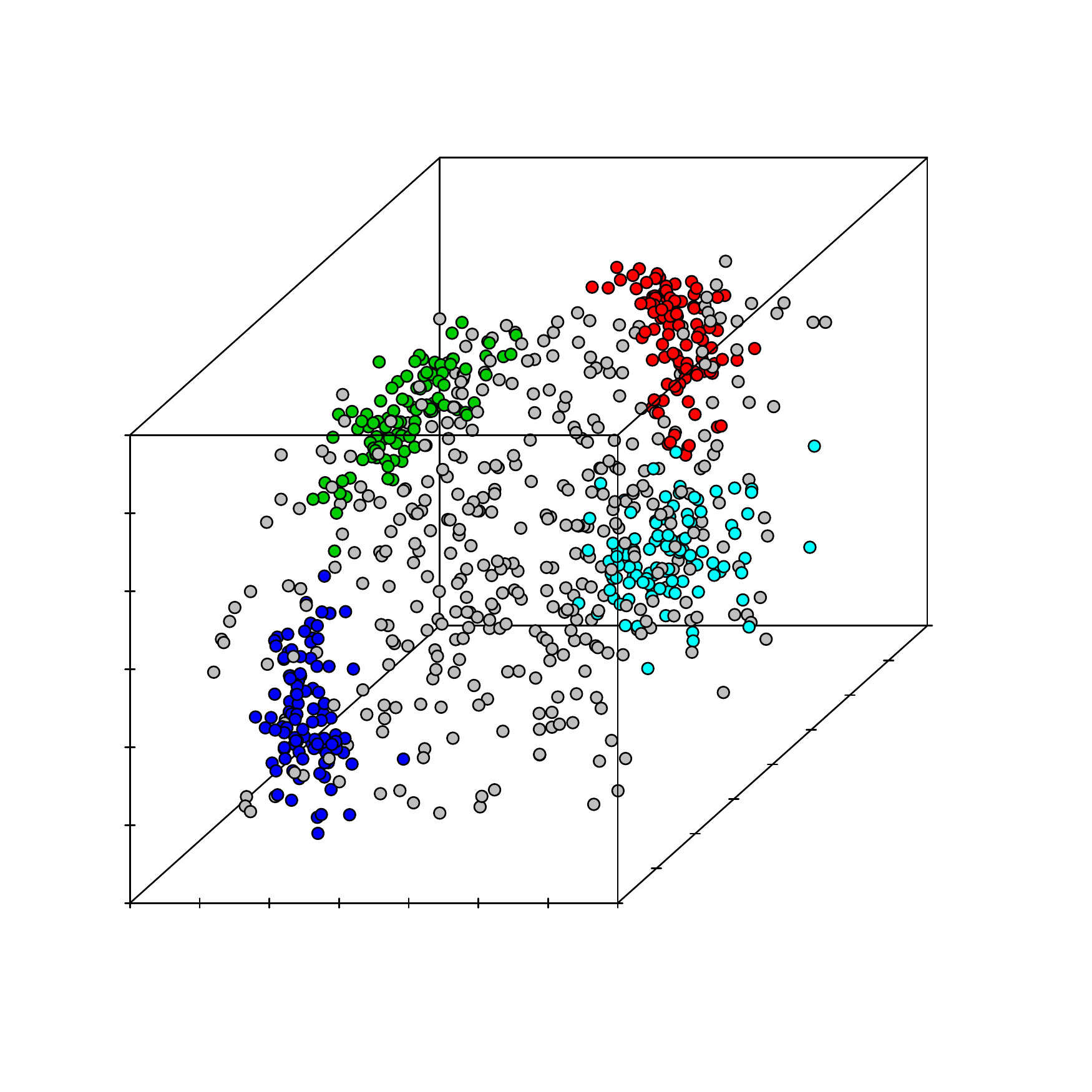} 
\end{tabular}
\end{center}
\begin{center}
\captionof{figure}{\em Left: Estimated clusters in the noisy
four crescents data. Right: Example of three dimensional data}
\label{fig::NEW}
\end{center}

The left plot in Figure \ref{fig::NEW} shows the identified clusters, 
in the noisy crescent data. Part of the random noise has been,
correctly, assigned to some of the four main clusters. 

The right panel in Figure~\ref{fig::NEW}
contains an example of a three dimensional data set, 
consisting of 400 data points, with four clusters, and 
400 points of uniform random noise. 

The procedure is 
unchanged when data are in more than two dimensions. 
This is clearly shown in the left panel of Figure 
\ref{fig::ThreeD1} where the mode diagram $\cal{\hat{D}}$ 
shows four points above the threshold function. The right 
panel in Figure \ref{fig::ThreeD1} presents the four 
estimated clusters where, as before,  points of 
random noise are assigned to the main clusters, according to 
their closeness to the modes.

\vspace{-.5in}

\begin{center}
\begin{tabular}{ll}
\includegraphics[scale=.35]{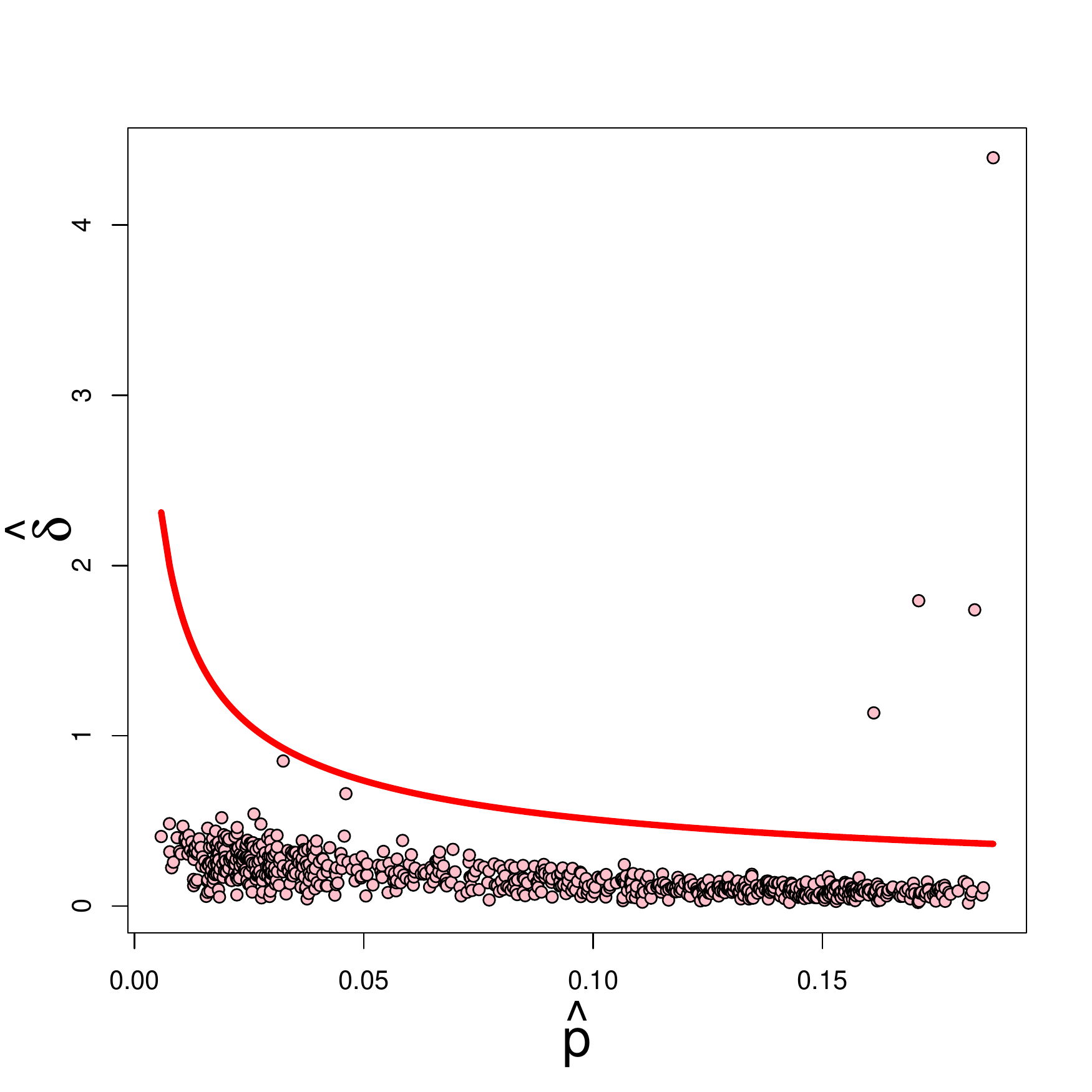} 
\includegraphics[scale=.35]{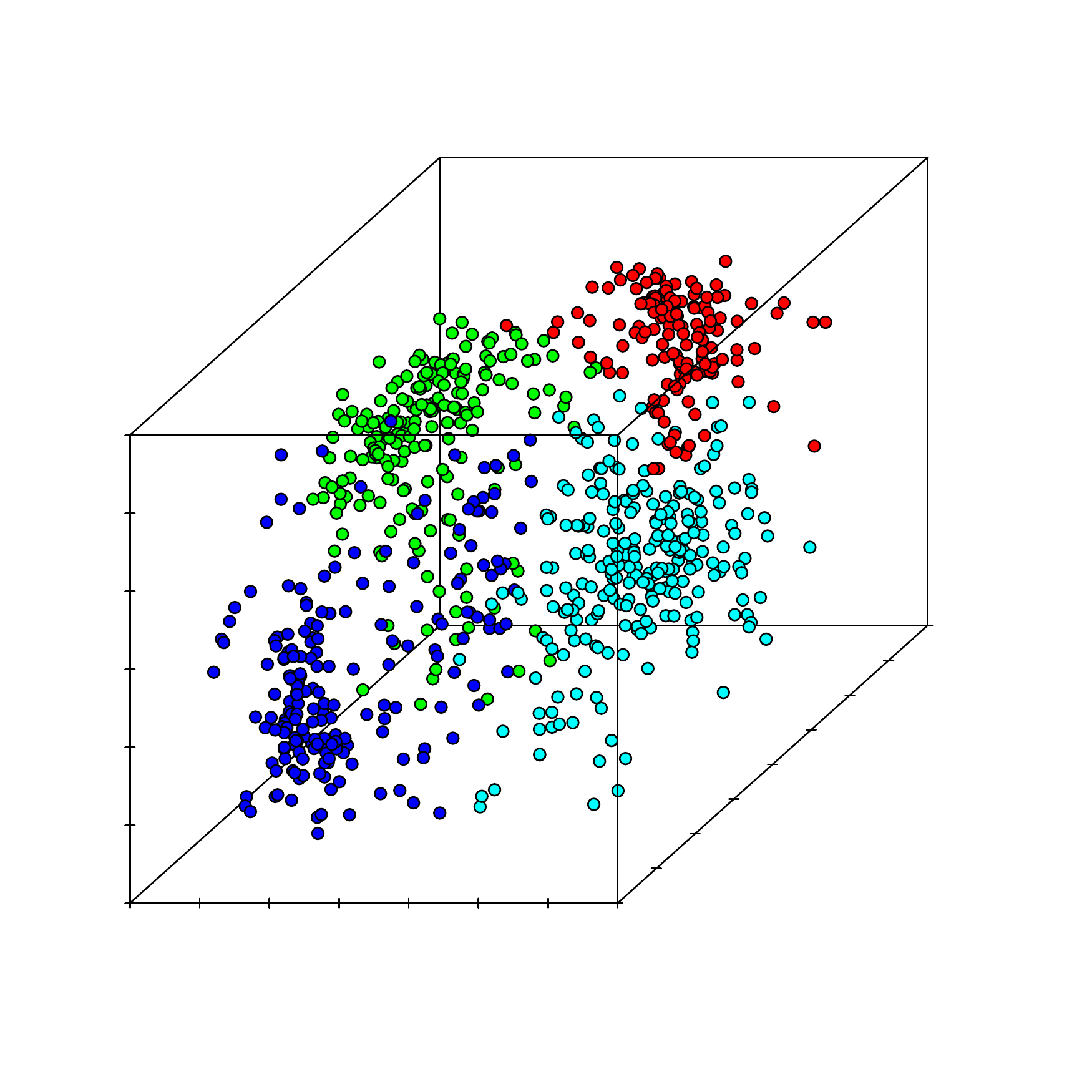} 
\end{tabular}
\end{center}
\begin{center}
\captionof{figure}{\em Three dimensional noisy data set.
Left: Estimated Mode diagram $\cal{\hat{D}}$, threshold 
function. Right: Estimated clusters}
\label{fig::ThreeD1}
\end{center}

\vspace{-.3in}
The mode diagram 
in Figure \ref{fig::ThreeD1}
includes two points with small values 
of both $\hat{\delta}$ and $\hat{p}$, that are just 
below the threshold function. In a different run of the code,
(Figure \ref{fig::ThreeD2}),
one of the points in the lower~left section of the diagram
is slightly above the threshold function. This would signal 
the existence of an extra cluster. But the right panel of
Figure \ref{fig::ThreeD2}, shows again only four clusters, 
meaning that the indication of the fifth mode is extremely 
weak. In fact for a sample point to~be a~mode,   
both values of $\hat{\delta}$ and $\hat{p}$ need 
to be large. In other runs of the code some of 
those points can be slightly farther 

\vspace{-.5cm}

\begin{center}
\begin{tabular}{ll}
\includegraphics[scale=.35]{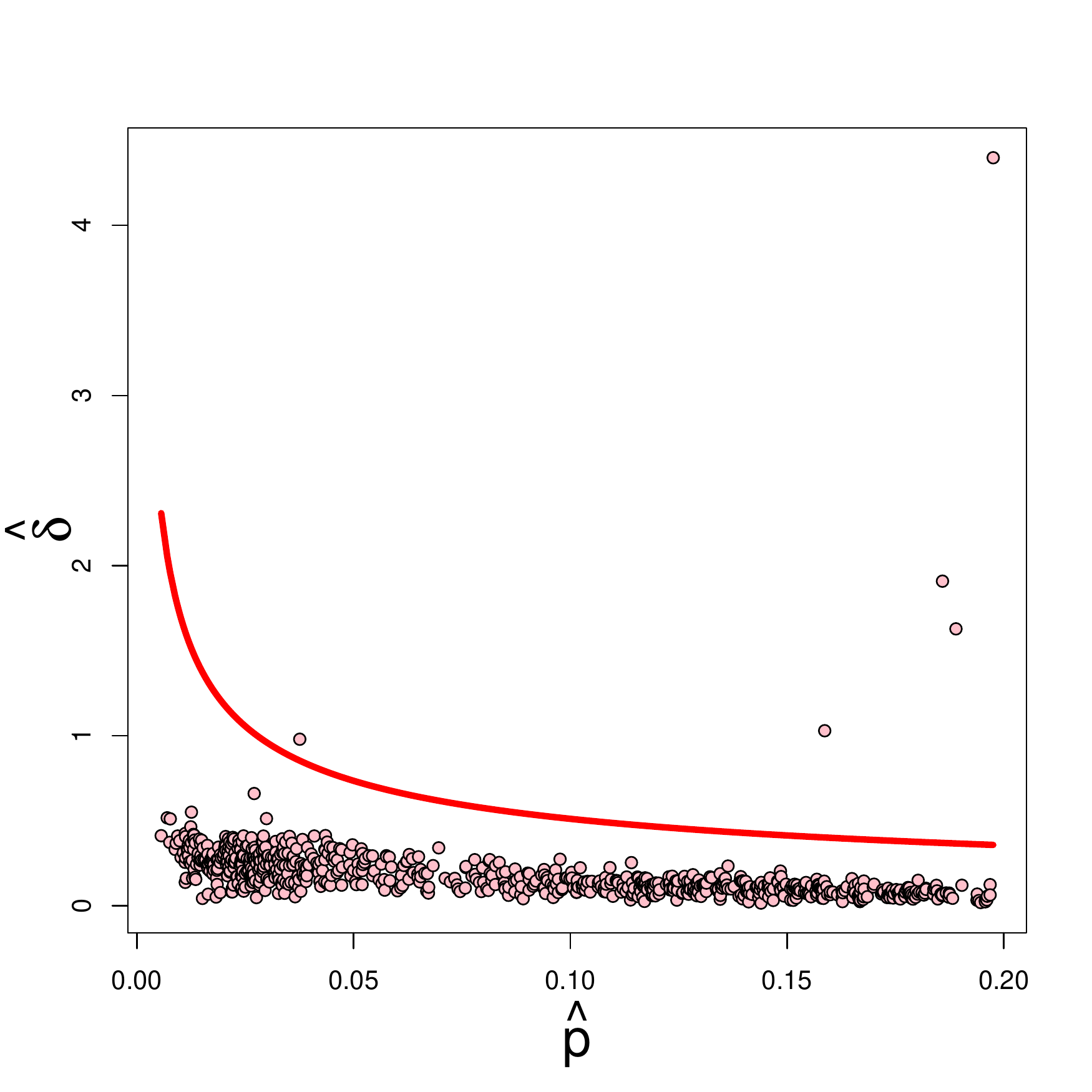} 
\includegraphics[scale=.35]{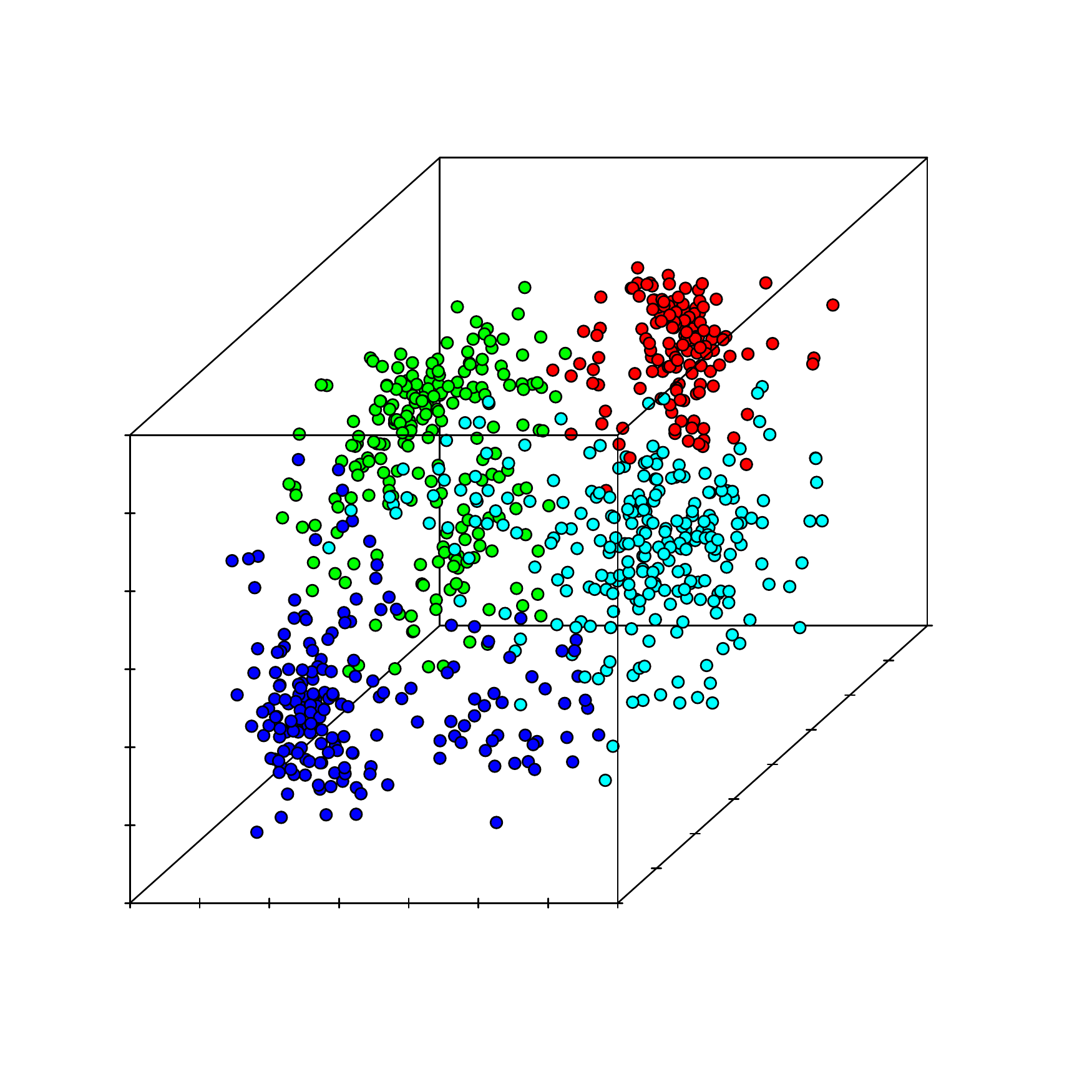} 
\end{tabular}
\end{center}
\begin{center}
\captionof{figure}{\em Three Dimensional Noisy data. 
Left: Estimated Mode diagram $\cal{\hat{D}}$, showing 
five modes above the threshold function. Right: Four 
estimated clusters}
\label{fig::ThreeD2}
\end{center}

\vspace{-.2in}

from the threshold function, so a fifth, or 
sixth cluster might be found, but they consist of a 
very small number of points. Thus, these extra clusters 
do not affect the qualitative results of our clustering 
procedure.

Our last example consists in $6,000$ data points in 
$15$-dimensions. Two multivariate Normal clusters have 
been included, with identity covariance matrix, and
mean vectors $-\mu${\bf 1}, and $\mu${\bf 1}. The plots 
in Figure \ref{fig::15Dim} show four estimates of 
$\cal{D}$, for increasing values of $\mu$.
The top left plot shows 
$\hat{\cal D}$ for $\mu = 0.5$. The clusters here are 
very close and
the mode diagram $\hat{\cal D}$ identifies only one 
mode presenting large values of both $\hat{p}$ and 
$\hat{\delta}$. The remaining points above 
the threshold indicate other possible modes, generated
by the closeness of the two clusters.
As the value of $\mu$ increases to $1, 3,$ and $10$,
the remaining three plots in Figure \ref{fig::15Dim}
show that the added distance between the modes in 
the data, identifies the presence of two separate clusters.

It is clear from this last example that
we could improve the clustering by increasing the choice of $M$
in (\ref{eqn::tn}).
In fact, from Figure \ref{fig::15Dim}
it seems that it would be simple to derive a data-adaptive
choice of $M$.
We could choose $M$ so that the number of points above $t_n$ remains stable.
This would correctly identify the two clusters in the cases with $\mu \geq 1$.
However, this idea is beyond the scope of the current paper.
More important is the fact that this simple two dimensional plot 
nicely summarizes the information in the 15-dimensional data.

\vspace{-.5in}

\begin{center}
\includegraphics[scale=.7]{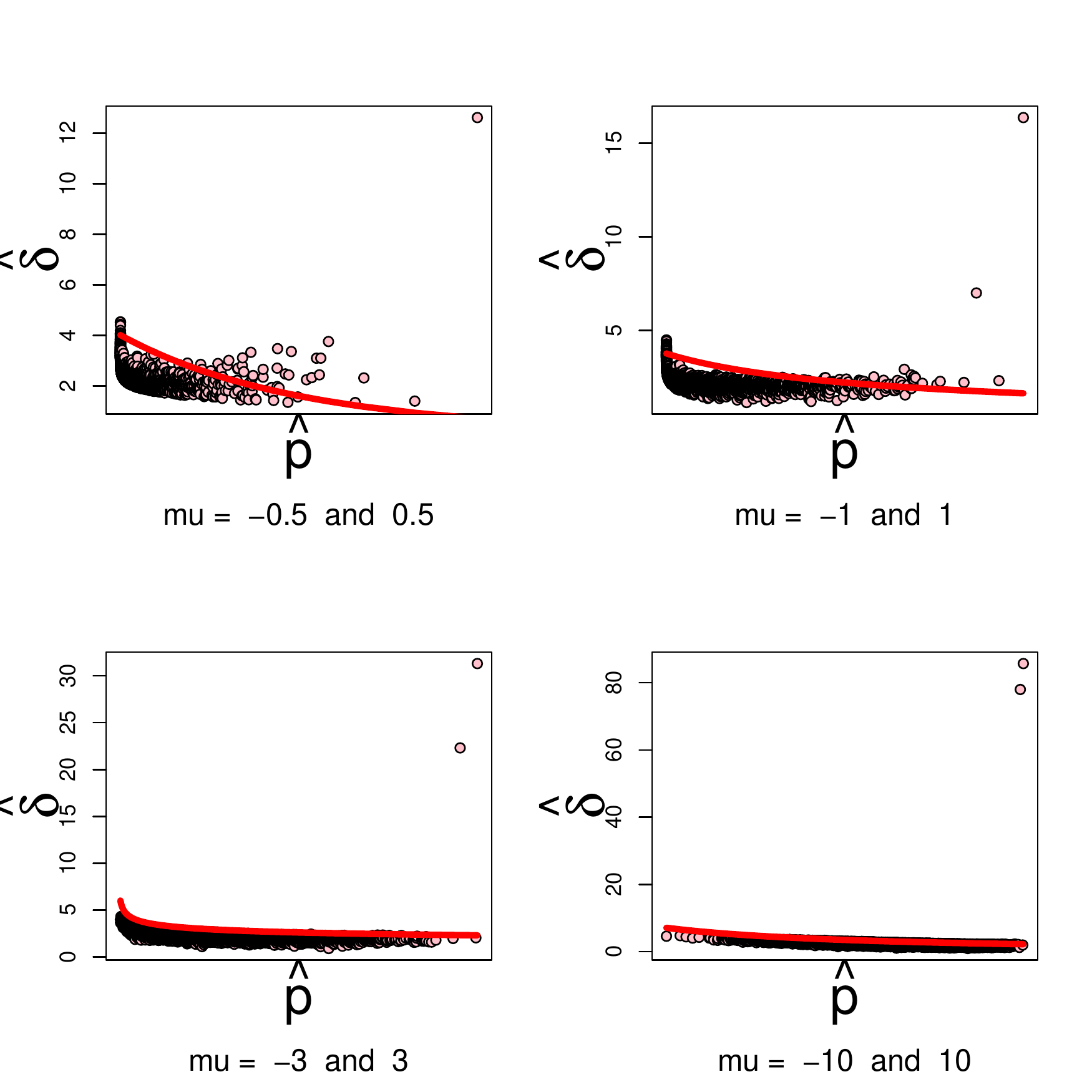} 
\end{center}
\begin{center}
\captionof{figure}{\em 15-dimensional data set. $6,000$
sample points. Plots of $\hat{D}$ for increasing distance 
between two clusters}
\label{fig::15Dim}
\end{center}

\vspace{-.5in}

\section{Conclusion}
\label{section::conclusion}

We have studied the properties of the mode diagram
introduced by
\cite{rodriguez2014clustering}.
We have seen that, for non-modes,
$\log \hat \delta(X_i)$ falls on or below a linear function
of $\log \hat p(X_i)$.
Based on this observation,
we suggested a robust regression method
for classifying points on the mode diagram as modes or non-modes.
We would like to emphasize that we think that the mode diagram
is a useful visualization method
for mode-based clustering
regardless of how one separates modes from non-modes.

Our analysis depended on a number of assumptions.
In particular, we assumed that the density $p$ is Morse.
This assumption is made --- explicitly or implicitly --- in most
density based clustering methods.
Loosely speaking, this means that $p$ has no ``flat regions.''
(A notable exception is \cite{jiang2016modal} who specifically allow for such flat region.)
We conjecture that the mode diagram still provides useful
information when $p$ is not Morse but proving this will require
new tools.

Although we have suggested a default value for $M$ that
defines the threshold functions,
it seems likely that it is possible to define a data-driven choice of $M$.
We hope to report on this in future work.

It would be interesting to develop a similar diagnostic plot
for other clustering methods such as $k$-means clustering.
Currently, we are not aware of any such diagnostic plots.

\section*{Appendix: Proofs}

{\bf Proof of Theorem \ref{thm::sample-mode}.}
{\em (i)} Let $X_{(j)}$ be the closest point to $m_j$.
From Lemma \ref{lemma::order-stat} below,
$||X_{(j)}-m_j|| \leq \epsilon_n$.
For any point $x\in B(m_j,\omega_j)$ 
we have
\begin{equation*}
p(x) = p(m_j)-\frac{1}{2}\;(x-m_j)^T \ \Biggl[\ \int_0^1 J(u \,m_j + (1-u\,)x)\ 
             du \ \Biggr]\ (x-m_j).
\end{equation*}
Thus for $X_j$
\begin{align}\label{eq::Taylor}
p(X_j) & = p(\,m_j) - \frac{1}{2} (X_j-m_j)^T 
    \left[\int_0^1 J(u, m_j + (1-u)\,X_j) \,  du \right] (X_j-m_j) \nonumber\\
     & \\
     & \leq p(m_j) - \frac{\lambda_j}{4} ||X_j-\!m_j||^2  \nonumber 
\end{align}
and for $X_{(j)}$
\begin{align*}
p(X_{(j)}) & =  p(\, m_j) -  \frac{1}{2}\,(X_{(j)}-m_j)^T 
        \left[\,\int_0^1 J(u \,m_j + (1-u)\,X_{(j)})\, du\ \right] (X_{(j)}-m_j)\\
          & \geq \,p(m_j) - \frac{\Lambda_j}{2} \; ||X_{(j)}-m_j||^2.
\end{align*}
Then
$$
p(\,m_j) - \frac{\lambda_j}{4} \; ||X_j-m_j||^2 \geq p(X_j) \geq p(X_{(j)}) \geq
p(\,m_j) - \frac{\Lambda_j}{2} \; ||X_{(j)}-m_j||^2
$$
which implies that
\begin{equation}\label{eq::Close}
||X_j-m_j||^2 \leq \frac{2\Lambda_j}{\lambda_j} \; ||X_{(j)} - m_j||^2 
           \leq \frac{2\Lambda_j}{\lambda_j}\; \epsilon_n^2 \leq \psi_{n,j}^2.
\end{equation}
This proves {\em(i)}.

{\em (ii)} Let $X_i$ be any point with $p(X_i) > p(X_j)$.
By definition, $X_i\notin B(m_j,\omega_j)$.
By the triangle inequality,
$$
\omega_j \leq ||X_i-m_j|| \leq ||X_i - X_j|| +  ||X_j-m_j|| \leq
||X_i-X_j|| + \sqrt{ \frac{2\Lambda_j}{\lambda_j}} \,\epsilon_n
$$
Thus, since $\epsilon_n \to 0$
$$
||X_i-X_j|| \geq \omega_j - \sqrt{\frac{2\Lambda_j}{\lambda_j}}\,\epsilon_n
\geq \omega_j/2,
$$
and $\delta(X_j) \geq \omega_j/2$. $\Box$

\begin{lemma}
\label{lemma::order-stat}
Let $X_{(j)}$ be the closest point to $m_j$ for $j=1,\ldots, k$.
Then
\begin{equation}
P^n\Bigl( \max_{1\leq j\leq k}||X_{(j)} - m_j|| > \epsilon_n\Bigr)\to 0.
\end{equation}
Hence,
with probability tending to 1,
each ball $B(m_j,\epsilon_n)$ contains at least one point.
\end{lemma}

\begin{proof}
Let $v_d$ be the volume of the unit ball, let 
$B = B(m_j,\epsilon_n)$. For a sequence of points
$y_{n,j}\in B$ converging to $m_j$ as $n\to\infty$
$$
P(B) = p(y_{n,j})\mu(B) = p(y_{n,j})\epsilon_n^d \;v_d
$$
Since $y_{n,j}\to m_j$ as $n\to\infty$, when $n$ is large
$p(y_{n,j}) > p(m_j)/2$, so
\begin{align*}
P^n(|| X_{(j)} - m_j|| > \epsilon_n) &=
P^n(||X_i-m_j|| > \epsilon_n\ {\rm for\ all\ }i) \\
& = [1-P(B)]^n \leq [1 - p(y_{n,j})\epsilon_n^d \,v_d]^n\\
& \leq
\exp\left\{ -n p(y_{n,j})\epsilon_n^d\, v_d\right\}\\
& = \left(\frac{1}{n}\right)^{ v_d \,p(y_{n,j})} 
\leq\left(\frac{1}{n}\right)^{ v_d \;p(m_j)/2}.
\end{align*}
Hence,
$$
P^n\Bigl( \max_{1\leq j\leq k}||X_{(j)} - m_j|| > \epsilon_n\Bigr) \leq 
\sum_{j=1}^k \left(\frac{1}{n}\right)^{ v_d \;p(m_j)/2} \to 0.
$$
\end{proof}

\bigskip
{\bf Proof of Theorem \ref{thm::oracle-theorem}.}
Let $t_n(u)$ be defined in \eqref{eq::tn-def}.

{\em (i)}  Let $X_j \in {\cal X}_1$ be a local mode.
From \eqref{eq::Taylor} and \eqref{eq::Close} it is immediate 
that $p(X_j) = p(m_j) + O_P(\psi_n^2)$. From Theorem 
\ref{thm::sample-mode}, 
$\delta(X_j)\geq \omega_j/2>0$,
then from the definition of $t_n(p(X_j))$ we have
$\frac{\delta(X_j)}{t_n(p(X_j))} \to \infty$.

{\em (ii)} Let $X_i\in {\cal X}_2$. 
Since $X_i \notin {\cal X}_1$ it is not a local mode.
So there exists a local mode 
$X_i\in B(m_j,\psi_{n,j})$ with $p(X_i) < p(X_j) \leq p(m_j)$.
Hence, $\delta(X_i) \leq \psi_{n,j}$.
So, from (\ref{eq::C}),
\begin{equation}\label{eq::delta2}
\delta^2(X_i)  \leq  \psi^2_{n,j} =
 \frac{2G^2 \Lambda_j \epsilon_n^2}{\lambda_j} \leq
 \left(\frac{C\log n}{n p(m_j)}\right)^{2/d} \leq
 \left(\frac{C\log n}{n p(X_i)}\right)^{2/d} =  t_n^2(p(X_i))
\end{equation}
as required.

\vspace{-.1in}

{\em (iii)} Let $x\in \Gamma^c$.
Recall that $p(x) \geq a >0$ for all $x$.
From (\ref{eq::delta2}), we see that
$t_n(x) \geq \psi_n = \max \psi_{n,j}$.

\newpage
So
\begin{align*}
P( \delta(x) > t_n(x)) & \leq
P( \delta(x) > \psi_n(x)) =
\prod_i P(X_i \notin B(x,\psi_n(x)) \bigcap L_x)\\
&=
\left[ 1- P(X_i \in B(x,\psi_n(x)) \bigcap L_x)\right]^n\\
& \leq
\left[ 1- a \mu( B(x,\psi_n(x)) \bigcap L_x)\right]^n\\
& \leq 
\left[ 1- a \tau \mu( B(x,\psi_n(x)) \right]^n\\
& \leq
\exp\left( - n a \tau v_d \psi_n^d \right) =
\exp\left( - n a \tau v_d 2^{d/2}G^d \left(\frac{\Lambda}{\lambda}\right)^{d} \frac{\log n}{n}\right)\\
& \leq
\exp\left( - 3\log n\right) = \left(\frac{1}{n}\right)^3
\end{align*}
where we used the fact that
$\mu( B(x,\psi_n(x)) \bigcap L_x) \geq \tau \mu( B(x,\psi_n(x)))$
due to (A4).
So,
$$
P\Bigl( \delta(X_i) > t_n(X_i)\ {\rm for\ some\ }X_i \in \Gamma^c\Bigr) \leq
n\left(\frac{1}{n}\right)^3 \to 0.\ \ \ \Box
$$

{\bf Proof of Lemma \ref{theorem::exponential}.}
Fix $s>0$ and let $B_n = B\left(x, \left(\frac{s}{n}\right)^{1/d}\right)$.
Define $A_x = B_n \bigcap \Bigl\{y:\ p(y) > p(x)\Bigr\}$.
There exists a sequence $y_n\to x$ such that
$P(A_x) = p(y_n)\mu(A_x)$.
From (A4'),
\begin{align*}
P(A_x) &= p(y_n)\mu(A_x) = p(y_n)\frac{\mu(A_x)}{\mu(B_n)} \mu(B_n) \\
&=
(p(x) + o(1)) (\tau + o(1)) \left(\frac{s}{n}\right)=
\frac{s p(x) \tau}{n} + o\left(\frac{1}{n}\right).
\end{align*}
Then 
\begin{align*}
P\left(n\; \delta(x)^d \leq s \right)
& = P\left[\delta(x)\leq(s/n)^{1/d}\right] 
  = 1-P\left[\delta(x)> (s/n)^{1/d}\right] \\
& = 1-\prod_{i=1}^n P\left[ X_i \notin A_x \;{\rm for\; all\;} X_i\right]
  = 1-\left[P(X \notin A_x)\right]^n = 1-\left[1-P(X\in A_x)\right]^n\\
& = 1- \exp\left\{n \log \left[1-P(X\in A_x)\right]\right\} 
  = 1- e^{-s\tau p(x)} e^{o(1)} \to    1- e^{-s\tau p(x)}.
\end{align*}
The final statement, about modes, follows since $\delta(x)$ is strictly positive.
$\Box$

{\bf Remark.}
If we had not assume that $p$ is bounded from below,
then one needs to work with the truncated region of the density 
$p(x) \geq a_n = n^{-1/(d+2)}$ and then 
$P^n(F_n^c) \leq N  [1- a\, v_d\, \epsilon_n^d]^n = 
\left[\frac{C_2}{\epsilon_n} \right]^d \; 
\left[1-n^{-\frac{1}{d+2}}\,v_d \epsilon_n  \right]^n \to 0$.

\bigskip

{\bf Proof of Theorem \ref{thm::about-threshold}.}
The proofs of (i) and (ii) mimic
the proof if Theorem 2, with $\hat p$ replacing $p$.
We focus on (iii).

In Lemma \ref{lemma::covering},
we show that there exists balls
$B(c_1,\epsilon_n),\ldots, B(c_N,\epsilon_n)$
such that
the support of $p$ is contained in
$\bigcup_{s=1}^N B(c_j,\epsilon_n)$ and such that,
$P^n (F_n)\to 1$
where $F_n$ is the event that each ball contains at least one data point.

Let 
$$
\Gamma = \bigcup_{j=1}^k B(\hat m_j,c_1 \epsilon_n).
$$
In Lemma \ref{lemma::killer}, we show that the following is true.
For every $x\in \Gamma^c$,
there exists a ball $B$ such that
(i) the radius of $B$ is $2\epsilon_n$,
(ii) $x\notin B$ but
(iii) $\min_{z\in B}||z-x|| \leq c_2 t_n(x)$ for some $c_2>0$ not depending on $x$.
Since this holds for all $x\in \Gamma^c$ it also holds for all
$X_i \in {\cal X}_3$.
So there is a ball $B_i$ such that
(i) the radius of $B_i$ is $2\epsilon_n$,
(ii) $X_i\notin B_i$ but
(iii) $\min_{z\in B}||z-X_i|| \leq c_2 t_n(X_i)$.
Now $B$ must contain at least one of the covering balls
$B(c_j,\epsilon_n)$.
On $F_n$,
this ball contains at least one point $X_j$, which is distinct from $X_i$.
It follows that
$\hat\delta(X_i) \leq c_2 t_n(X_i)$.
As this holds simultaneously for all
$X_i \in {\cal X}_3$, the result follows. $\Box$

\begin{lemma}
\label{lemma::covering}
There exists a set ${\cal B} =\{B_1,\ldots, B_N\}$
where each $B_j$ is a ball of radius $\epsilon_n$,
$N = \left[\xi/\epsilon_n\right]^d$ for some $\xi>0$,
${\cal X} \subset \bigcup_j B_j$. Let $F_n$ denote the event
that each ball contains at least one data point. Then
$P^n(F_n)\to 1.$
\end{lemma}

{\bf Proof of Lemma \ref{lemma::covering}.}
Since ${\cal X}$ is a compact subset of $\mathbb{R}^d$,
there exists a covering
${\cal B}=\{B_1,\ldots, B_N\}$
of the sample space with balls of size
$\epsilon_n$
where $N = \left[\xi/\epsilon_n\right]^d$ for some $\xi>0$.
Let $x_j$ denote the center of $B_j$.
Note that $P(X\in B_j) \geq a\, \epsilon_n^d v_d$
where $v_d$ is the volume of the unit ball.
Then
\begin{align*}
P^n(F_n^c) &= P^n({\rm some\ }B_j\ {\rm is\ empty}) \leq
\sum_{j=1}^N P^n(B_j\ {\rm is\ empty})\\
&=
\sum_{j=1}^N \prod_{i=1}^n P(X_i \notin B_j) =
\sum_{j=1}^N \prod_{i=1}^n [1-P(X_i \in B_j)]\\
& \leq
\sum_{j=1}^N \prod_{i=1}^n [1 - a\, v_d \,\epsilon_n^d]=
\sum_{j=1}^N  [1- a\, v_d\, \epsilon_n^d]^n = 
N  [1- a\, v_d\, \epsilon_n^d]^n \\
& \leq
N e^{-n \,a \,v_d \epsilon_n^d}=
\frac{\xi^d \;n}{\log n} \left(\frac{1}{n}\right)^{a \,v_d\,r} \to 0
\end{align*}
since $r > 1/(av_d)$.
Hence $P^n(F_n)\to 1$. $\Box$

\begin{lemma}
\label{lemma::killer}
Let $p$ be a Morse function
with finitely many critical points
and modes
${\cal M} = \{m_1,\ldots, m_k\}$.
Let $A_s$ denote the supremum of the $s^{\rm th}$ derivative of $p$
for $s=0,1,2,3$.
(Hence, $A_0 = \sup_x p(x)$.)
There exists positive constants $c_1$ and $c_2$ such that
the following is true.
For every $x\in \Biggl(\bigcup_{j=1}^k B(m_j,c_1 \epsilon)\Biggr)^c$,
there exists a ball $B$ of radius $2\epsilon_n$ such that:
\begin{enum}
\item $\max_{z\in B} ||z-x|| \leq c_2 t_n(x)$ and
\item $x\notin B$,
\item $p(z) > p(z)$ for all $z\in B$.
\end{enum}
\end{lemma}

{\bf Proof of Lemma \ref{lemma::killer}.}
Let $S=\{s_1,\ldots, s_N\}$ be the set of critical points that are not modes.
Let $u_n = c_1 \epsilon_n/2$.

{\bf Case 1:}
Suppose that
$||x-s_j||\leq u_n$ for some $s_j\in S$.
Note that $s_j$ cannot be a mode since
$u_n < c_1 \epsilon_n$.
Let $\lambda_j$ be the largest eigenvalue of $H(s_j)$ and note
that $\lambda_j>0$ since $s_j$ is not a mode.
Let $v$ be the corresponding eigenvector and define
$B \equiv B(y,2\epsilon_n)$ where
$$
y = x + c_3 \epsilon_n v
$$
where $c_3 \equiv c_3(x)$ is such that
\begin{equation}\label{eq::c3}
\max\Biggl\{2, 
4c_1 \sqrt{2}q_j/\lambda_j,\ 
\sqrt{ \frac{16}{\lambda_j}[\sqrt{2}q_j + 4 A_2]},\ 
\frac{32 A_2}{\lambda_j}\Biggr\} < c_3 < c_2 \left(\frac{C}{A_0}\right)^{1/d}-2.
\end{equation}
Here, $c_1$ and $c_2$ are any positive constants such that
the above interval is nonempty.
Let $q^2_j$ be the largest eigenvalue of $H^2(s_j)/2$.
Because $p$ is Morse, $q_j^2 >0$.

1. Let $z\in B$.
Then
$$
||z-x|| \leq ||z-y|| + ||y-x|| \leq 2\epsilon_n + c_3 \epsilon_n \leq c_2 t_n(x)
$$
where we used the fact 
(from (\ref{eq::c3})) that
$c_3 \leq c_2 (C/A_0)^{1/d} -2$
which implies that
$c_3 \leq c_2 (p(x)/A_0)^{1/d} -2$
and hence
$(2+c_3)\epsilon_n \leq c_2 t_n(x)$.

2. Next, note that
$$
\min_{z\in B} ||z-x|| = ||x-y|| - 2\epsilon_n = c_3 \epsilon_n - 2\epsilon_n >0
$$
since $c_3 > 2$.
Hence, $x\notin B$.

3. For all $0\leq r\leq 1$,
$v^T H(r y + (1-r)x) v \geq v^T H(s_j)v - O(\epsilon_n) \geq \lambda_j/2$ for all large $n$.
So
\begin{align}\nonumber
p(y) &=  p(x) + c_3 \epsilon_n v^T g(x) + \frac{c_3^2 \epsilon_n^2}{2} 
v^T \int_0^1 H(s y + (1-s)x) ds \ v \\
&\geq  p(x) + c_3 \epsilon_n v^T g(x) + \frac{c_3^2 \epsilon_n^2\lambda_j}{4} \geq
p(x) + 
\frac{c_3^2 \epsilon_n^2\lambda_j}{4} - c_3 \epsilon_n |v^T g(x)|.
\label{eq::py}
\end{align}
Now
$$
g(x)= g(s_j) + H(s_j)(x-s_j) +R_n=  H(s_j)(x-s_j) +R_n
$$
where the norm of the remainder $R_n$ is bounded by
$\sqrt{d}A_3 u_n^2$.
So, for all large $n$,
$$
g(x)^T g(x) = (x-s_j)^T H^2(s_j) (x-s_j) + O(u_n^3)\leq 
u_n^2 q_j^2 + O(u_n^3) \leq 2u_n^2 q_j^2.
$$
So $||g(x)|| \leq  \sqrt{2} q_j u_n$ and so
$c_3 \epsilon_n |v^T g(x)| \leq c_1 c_3 \sqrt{2}q_j \epsilon_n^2/2$
and from
(\ref{eq::py}) we conclude that
\begin{equation}\label{eq::py2}
p(y) \geq
p(x) -
\frac{c_1 c_3 \sqrt{2}q_j \epsilon_n^2}{2} +
 \frac{c_3^2 \epsilon_n^2\lambda_j}{4}  \geq p(x) + \frac{c_3^2 \epsilon_n^2 \lambda_j}{8}
\end{equation}
since
$c_3 \geq 8 c_1 \sqrt{2}q_j/(2\lambda_j)$.

Now consider any $z\in B$.
Then
$$
p(z) -p(x) = p(z) - p(y) + p(y) - p(x) \geq p(z) - p(y) + \frac{c_3^2 \epsilon_n^2 \lambda_j}{8}.
$$
We have
\begin{align*}
p(z) &=
p(y) + (z-y)^T g(y) + R_n >
p(x) + \frac{c_3^2 \epsilon_n^2 \lambda_j}{8}  + (z-y)^T g(y) + R_n\\
&=
p(x) + \frac{c_3^2 \epsilon_n^2 \lambda_j}{8} + (z-y)^T [g(x) + \tilde R_n] + R_n\\
&=
p(x) + \frac{c_3^2 \epsilon_n^2 \lambda_j}{8} + (z-y)^T g(x) + (z-y)\tilde R_n + R_n
\end{align*}
where
$$
|(z-y)^T g(x)| \leq ||z-y||\ ||g(x)|| \leq 2\epsilon_n \sqrt{2} q u_n = \sqrt{2} \epsilon_n^2 q,
$$
$$
|R_n| \leq ||z-y||^2 A_2 \leq 4\epsilon_n^2 A_2
$$
and
$$
||(z-y)\tilde R_n|| \leq ||z-y|| \ ||y-x|| A_2 \leq
(2 \epsilon_n) (c_3 \epsilon_n ) A_2 = 2 c_3 \epsilon_n^2 A_2
$$
so that
$$
p(z) > p(x) + 
\frac{c_3^2 \epsilon_n^2 \lambda_j}{8}-
\sqrt{2} \epsilon_n^2 q - 4\epsilon_n^2 A_2-2 c_3 \epsilon_n^2 A_2> p(x)
$$
since
$$
c_3 > \max\Biggl\{\sqrt{ \frac{16}{\lambda_j}[\sqrt{2}q_j + 4 A_2]},\ 
\frac{32 A_2}{\lambda_j}\Biggr\}.
$$

\vspace{1cm}

{\bf Case 2}.
Suppose that $||x-s_j||> u_n$ for all $s_j$.
First, we will need to lower bound $||g(x)||$.
By definition,
$$
x\in \Biggl[\Biggl(\bigcup_r B(s_j,u_n)\Biggr) \bigcup \Biggl(\bigcup B(m_j,c_1 \epsilon_n) \Biggr)\Biggr]^c.
$$
For all large $n$,
the minimum of $||g(x)||$ over this set occurs
at the boundary of one of these balls.
That is
$$
||g(x)|| \geq
\min
\Biggl\{
\min_{s_j\in S}\inf_{w\in \partial B(s_j,u_n)}||g(w)||,\ \ 
\min_{m_j\in {\cal M}}\inf_{w\in \partial B(m_j,c_1\epsilon_n)}||g(w)||\Biggr\}.
$$
Using a Taylor expansion of $g(w)$ as in Case 1,
we then have that
\begin{equation}\label{eq::ming}
||g(x)|| > \min_j c_1 \epsilon_n q_j/4 = c_1 \epsilon_n q/4.
\end{equation}

Choose $c_3 \equiv c_3(x)$ such that
\begin{equation}\label{eq::c3a}
8 \left(\frac{A_0}{C}\right)^{1/d} < c_3 <
\min\Biggl\{
c_2 - 2 \left(\frac{A_0}{C}\right)^{1/d},\ 
c_1 \left(\frac{a}{C}\right)^{1/d} \frac{q}{4A_2} \Biggr\}.
\end{equation}
Here, $c_1$ and $c_2$ are any positive constants such that
the interval is nonempty.
Let 
$B = B(y,2\epsilon_n)$ where
$$
y = x + \frac{c_3 t_n(x) g(x)}{||g(x)||}.
$$
So $||y-x|| = c_3 t_n(x)$.

1. Let $z\in B$.
Then
$$
||z-x|| \leq ||z-y|| + ||y-x|| \leq 2\epsilon_n + c_3 t_n(x) \leq c_2 t(x)
$$
since
$c_3 \geq c_2 - 2(A_0/C)^{1/d}$.

2. Next
$$
\min_{z\in B} ||z-x|| = ||x-y|| - 2\epsilon_n = c_3 t_n(x) - 2\epsilon_n >0
$$
since
$c_3 > 2 (A_0/C)^{1/d}$. So $x\notin B$.

3. First we note that
$$
p(y) = p(x) + (y-x)^T g(x) + R_n =
p(x) + c_3 t_n(x) ||g(x)|| + R_n
$$
where
$$
|R_n| \leq ||y-x||^2 A_2 /2 = c_3^2 t_n^2(x) A_2/2 < \frac{c_3 t_n(x) ||g(x)||}{2}
$$
since
$||g(x)|| > c_1 \epsilon_n q/4$ and 
$c_3 < (a/C)^{1/d} c_1 q/(4A_2)$.
So
$$
p(y) > p(x) + \frac{c_3 t_n(x) ||g(x)||}{2}.
$$
Now consider any $z\in B$.
Then
$$
p(z) -p(x) =
p(z) - p(y) + p(y) - p(x) >
p(z) - p(y) + \frac{c_3 t_n(x) ||g(x)||}{2}.
$$
Now
$$
p(z) - p(y) = (z-y)^T g(y) + O(\epsilon_n^2)
$$
and so
$$
|p(z) - p(y)| \leq 2\epsilon_n ||g(y)|| +O(\epsilon_n^2) \leq 2\epsilon_n ||g(x)|| + O(\epsilon^2)
$$
and this
$p(z) > p(x)$ since
$$
\frac{c_3 t_n(x) ||g(x)||}{2} > 2\epsilon_n ||g(x)||
$$
since
$c_3 > 8 \left(\frac{A_0}{C}\right)^{1/d}$. $\Box$

\bibliographystyle{plainnat}
\bibliography{paper}

\begin{thebibliography}{15}
\providecommand{\natexlab}[1]{#1}
\providecommand{\url}[1]{\texttt{#1}}
\expandafter\ifx\csname urlstyle\endcsname\relax
  \providecommand{\doi}[1]{doi: #1}\else
  \providecommand{\doi}{doi: \begingroup \urlstyle{rm}\Url}\fi

\bibitem[Arias-Castro et~al.(2015)Arias-Castro, Mason, and
  Pelletier]{arias2015estimation}
Ery Arias-Castro, David Mason, and Bruno Pelletier.
\newblock On the estimation of the gradient lines of a density and the
  consistency of the mean-shift algorithm.
\newblock \emph{Journal of Machine Learning Research}, 2015.

\bibitem[Chac{\'o}n(2012)]{chacon}
Chac{\'o}n.
\newblock Clusters and water flows: a novel approach to modal clustering
  through morse theory.
\newblock \emph{arXiv preprint arXiv:1212.1384}, 2012.

\bibitem[Chac{\'o}n et~al.(2013)Chac{\'o}n, Duong, et~al.]{chacon2013data}
Jos{\'e}~E Chac{\'o}n, Tarn Duong, et~al.
\newblock Data-driven density derivative estimation, with applications to
  nonparametric clustering and bump hunting.
\newblock \emph{Electronic Journal of Statistics}, 7:\penalty0 499--532, 2013.

\bibitem[Chac{\'o}n et~al.(2015)]{chacon2015population}
Jos{\'e}~E Chac{\'o}n et~al.
\newblock A population background for nonparametric density-based clustering.
\newblock \emph{Statistical Science}, 30\penalty0 (4):\penalty0 518--532, 2015.

\bibitem[Chazal et~al.(2017)Chazal, Fasy, Lecci, Michel, Rinaldo, and
  Wasserman]{chazal2017robust}
Fr{\'e}d{\'e}ric Chazal, Brittany~T Fasy, Fabrizio Lecci, Bertrand Michel,
  Alessandro Rinaldo, and Larry Wasserman.
\newblock Robust topological inference: Distance to a measure and kernel
  distance.
\newblock \emph{To Appear: Journal of Machine Learning Research}, 2017.

\bibitem[Cheng(1995)]{cheng1995mean}
Yizong Cheng.
\newblock Mean shift, mode seeking, and clustering.
\newblock \emph{Pattern Analysis and Machine Intelligence, IEEE Transactions
  on}, 17\penalty0 (8):\penalty0 790--799, 1995.

\bibitem[Comaniciu and Meer(2002)]{comaniciu2002mean}
Dorin Comaniciu and Peter Meer.
\newblock Mean shift: A robust approach toward feature space analysis.
\newblock \emph{Pattern Analysis and Machine Intelligence, IEEE Transactions
  on}, 24\penalty0 (5):\penalty0 603--619, 2002.

\bibitem[Courjault-Rad{\'e} et~al.(2016)Courjault-Rad{\'e}, D'Estampes, and
  Puechmorel]{courjault2016improved}
Vincent Courjault-Rad{\'e}, Ludovic D'Estampes, and St{\'e}phane Puechmorel.
\newblock Improved density peak clustering for large datasets.
\newblock 2016.

\bibitem[Du et~al.(2016)Du, Ding, and Jia]{du2016study}
Mingjing Du, Shifei Ding, and Hongjie Jia.
\newblock Study on density peaks clustering based on k-nearest neighbors and
  principal component analysis.
\newblock \emph{Knowledge-Based Systems}, 99:\penalty0 135--145, 2016.

\bibitem[Genovese et~al.(2016)Genovese, Perone-Pacifico, Verdinelli, and
  Wasserman]{genovese2016non}
Christopher~R Genovese, Marco Perone-Pacifico, Isabella Verdinelli, and Larry
  Wasserman.
\newblock Non-parametric inference for density modes.
\newblock \emph{Journal of the Royal Statistical Society: Series B (Statistical
  Methodology)}, 78\penalty0 (1):\penalty0 99--126, 2016.

\bibitem[Jiang and Kpotufe(2016)]{jiang2016modal}
Heinrich Jiang and Samory Kpotufe.
\newblock Modal-set estimation with an application to clustering.
\newblock \emph{arXiv preprint arXiv:1606.04166}, 2016.

\bibitem[Li et~al.(2007)Li, Ray, and Lindsay]{li2007nonparametric}
Jia Li, Surajit Ray, and Bruce~G Lindsay.
\newblock A nonparametric statistical approach to clustering via mode
  identification.
\newblock \emph{Journal of Machine Learning Research}, 8\penalty0
  (Aug):\penalty0 1687--1723, 2007.

\bibitem[Milnor(2016)]{milnor2016morse}
John Milnor.
\newblock \emph{Morse Theory.(AM-51)}, volume~51.
\newblock Princeton university press, 2016.

\bibitem[Rodriguez and Laio(2014)]{rodriguez2014clustering}
Alex Rodriguez and Alessandro Laio.
\newblock Clustering by fast search and find of density peaks.
\newblock \emph{Science}, 344\penalty0 (6191):\penalty0 1492--1496, 2014.

\bibitem[Wang and Xu(2017)]{wang2015fast}
Xiao-Feng Wang and Yifan Xu.
\newblock Fast clustering using adaptive density peak detection.
\newblock \emph{Statistical methods in medical research}, pages 2800--2811,
  2017.

\end{thebibliography}

\end{document}